\documentclass{article}

\usepackage{arxiv}
\usepackage{amsmath,amssymb,amsfonts}
\usepackage{algorithmic}
\usepackage{url}
\usepackage{mathtools}
\usepackage{enumerate}
\usepackage{graphicx}
\usepackage{textcomp}
\usepackage{subfigure}
\usepackage{multirow}
\usepackage{amsmath}
\usepackage{hyperref}
\usepackage[utf8]{inputenc} % allow utf-8 input
\usepackage[T1]{fontenc}    % use 8-bit T1 fonts
\usepackage{hyperref}       % hyperlinks
\usepackage{url}            % simple URL typesetting
\usepackage{amssymb}
\usepackage{latexsym}

\usepackage{booktabs}       % professional-quality tables
\usepackage{amsfonts}       % blackboard math symbols
\usepackage{nicefrac}       % compact symbols for 1/2, etc.
\usepackage{microtype}      % microtypography
\usepackage{lipsum}
\usepackage{graphicx}
\usepackage{hyperref}
\usepackage{amsthm}
\usepackage[ruled]{algorithm2e}

\newtheorem{prop}{Proposition}
\graphicspath{ {./images/} }

\title{Improved gradient descent-based chroma subsampling method for color images in VVC}

\author{
 Kuo-Liang Chung \\
  Department of Computer Science and Information Engineering\\
  National Taiwan University of Science and Technology\\
  No. 43, Section 4, Keelung Road, Taipei, 10672, Taiwan, R.O.C. \\
  \texttt{klchung01@gmail.com} \\
   \And
 Szu-Ni Chen \\
  Department of Computer Science and Information Engineering\\
  National Taiwan University of Science and Technology\\
  No. 43, Section 4, Keelung Road, Taipei, 10672, Taiwan, R.O.C. \\
  \And
 Yu-Ling Lee \\
  Department of Computer Science and Information Engineering\\
  National Taiwan University of Science and Technology\\
  No. 43, Section 4, Keelung Road, Taipei, 10672, Taiwan, R.O.C. \\
  \And
 Chao-Liang Yu \\
  Department of Computer Science and Information Engineering\\
  National Taiwan University of Science and Technology\\
  No. 43, Section 4, Keelung Road, Taipei, 10672, Taiwan, R.O.C. \\
}

\begin{document}
\maketitle
\begin{abstract}
 Prior to encoding color images for RGB full-color, Bayer color filter array (CFA), and digital time delay integration (DTDI) CFA images, performing chroma subsampling on their converted chroma images is necessary and important.
 In this paper, we propose an effective general gradient descent-based chroma subsampling method for the above three kinds of color images, achieving substantial quality and quality-bitrate tradeoff improvement of the reconstructed color images when compared with the related methods.
First, a bilinear interpolation based 2$\times$2 $t$ ($\in \{RGB, Bayer, DTDI\}$) color block-distortion function is proposed at the server side, and then in real domain, we prove that our general 2$\times$2 $t$ color block-distortion function is a convex function.
Furthermore, a general closed form is derived to determine the initially subsampled chroma pair for each 2$\times$2 chroma block. Finally, an effective iterative method is developed to improve the initially subsampled $(U, V)$-pair.
Based on the Kodak and IMAX datasets, the comprehensive experimental results demonstrated that on the newly released versatile video coding (VVC) platform VTM-8.0, for the above three kinds of color images, our chroma subsampling method clearly outperforms the existing chroma subsampling methods.
\end{abstract}

% keywords can be removed

\section{Introduction}\label{sec:I}
\begin{figure*}[t]
  \centering
    \includegraphics[scale=0.4]{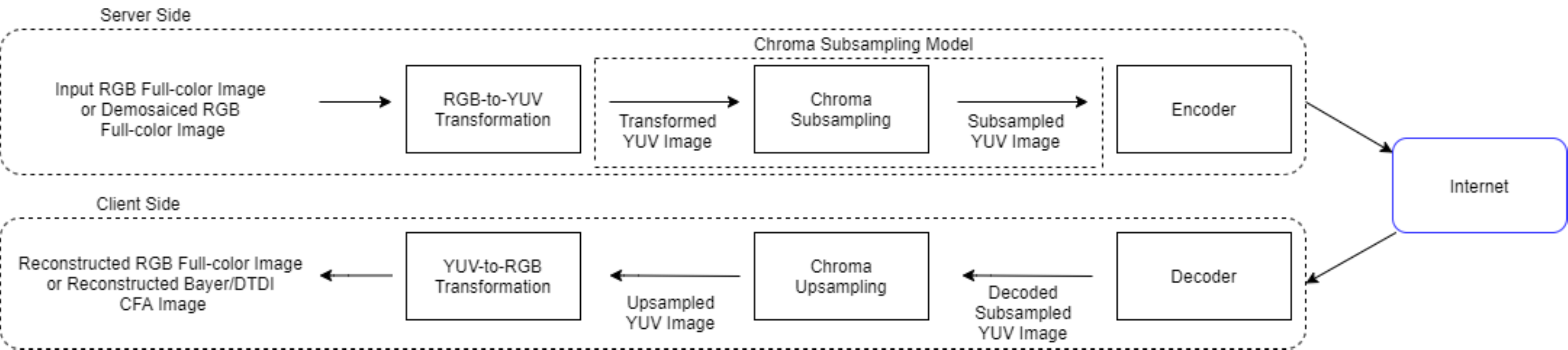}
  \caption{The chroma subsampling model in the coding system.}
  \label{pic:flow chart}
\end{figure*}
 In this study, we consider the chroma subsampling problem for the three kinds of color images, namely the RGB full-color image $I^{RGB}$, the Bayer color filter array (CFA) image $I^{Bayer}$ \cite{B.Bayer} which will be introduced in Subsection \ref{i_bayer}, and the digital time delay integration (DTDI) CFA image $I^{DTDI}$ \cite{E.Bodenstorfer} which will be introduced in Subsection \ref{i_dtdi}. For convenience, let the above three kinds of color images be denoted by the set {\bf CI} = \{$I^{RGB}$, $I^{Bayer}$, $I^{DTDI}$\}.
%Except for $I^{RGB}$, prior to performing chroma subsampling, $I^{Bayer}$ is first demosaicked to a RGB full-color image \cite{L. Zhang}, \cite{D. Kiku}; $I^{DTDI}$ is first demosaicked to a RGB full-color image \cite{DTDI}.

Before encoding the color image $\in$ {\bf CI} = \{$I^{RGB}$, $I^{Bayer}$, $I^{DTDI}$\}, except $I^{RGB}$, $I^{Bayer}$ and $I^{DTDI}$ should be demosaicked to RGB full-color images.
To demosaick $I^{Bayer}$ to a RGB full-color image, several demosaicking methods \cite{L.Zhang}, \cite{D.Kiku}, \cite{Li-2008} can be used, and in our study, the demosaicking method proposed by Kiku {\it et al.} \cite{D.Kiku} is used.
To demosaick $I^{DTDI}$ to a RGB full-color image, based on Kiku {\it et al.}'s residual interpolation idea, we developed a demosaicking method whose execution code can be accessed from the website \cite{DTDI}. Furthermore, following BT.601 \cite{BT601},
the RGB full-color image is transformed into a YUV image $I^{YUV}$ by using the following RGB-to-YUV conversion formula: 
\begin{equation}
\label{eq:RGB2YUV}
  \begin{small}
  \begin{bmatrix} Y_i \\ U_i \\ V_i \end{bmatrix} =
  \begin{bmatrix} 0.257 & 0.504 & 0.098 \\
                -0.148 & -0.291 & 0.439 \\
                  0.439 & -0.368 & -0.071 \end{bmatrix}
  \begin{bmatrix} R_i \\ G_i \\ B_i \end{bmatrix} +
  \begin{bmatrix} 16 \\ 128 \\ 128 \end{bmatrix}
  \end{small}
\end{equation}
in which $(R_i, G_i, B_i)$ and $(Y_i, U_i, V_i)$ indicate the RGB pixel and the converted YUV pixel triple-values, respectively, at the location $i$, $1 \leq i\leq 4$, according to the zigzag order in each 2$\times$2 RGB and YUV block-pair. It is noticeable that all discussion in this paper can be applied to the YCbCr color space because the color conversions, RGB-to-YCbCr and YCbCr-to-RGB,
are also linear as the color conversions, RGB-to-YUV and YUV-to-RGB.

In the chroma subsampling model, which is framed by a dotted box at the server side in Fig. \ref{pic:flow chart}, for 4:2:0, it determines one subsampled $(U, V)$-pair for each 2$\times$2 UV block; for 4:2:2, it determines one subsampled $(U, V)$-pair for each row of the 2$\times$2 UV block. 4:2:0 has been widely used in Bluray discs (BDs) and digital versatile discs (DVDs) for storing movies, sports, and TV shows. Throughout this paper, our discussion focuses on 4:2:0, although it is also applicable to 4:2:2.

After encoding the subsampled YUV image, which consists of the whole luma image and the subsampled chroma image, the encoded bit-stream is transmitted to the decoder via the internet. At the client side, which is shown in the lower part of Fig. \ref{pic:flow chart}, the decoded subsampled UV image is thus upsampled. Furthermore, the upsampled YUV image is transformed to the reconstructed RGB full-color image by the following YUV-to-RGB conversion:
\begin{equation}
  \label{eq:YUV2RGB}
  \begin{small}
  \begin{bmatrix} R_i \\ G_i \\ B_i \end{bmatrix} =
    \begin{bmatrix} 1.164 & 0 & 1.596 \\ 1.164 & -0.391 & -0.813 \\ 1.164 & 2.018 & 0 \end{bmatrix}
      \begin{bmatrix} Y_i-16 \\ U_i-128 \\ V_i-128 \end{bmatrix}
  \end{small}
\end{equation}
When the input image is $I^{Bayer}$ or $I^{DTDI}$, by Eq. (\ref{eq:YUV2RGB}), the upsampled YUV images can be directly transformed into the reconstructed Bayer and DTDI CFA images, respectively.

\subsection{Related Chroma Subsampling Work for Color Images} \label{sec:IA}
In this subsection, we introduce the related chroma subsampling work for $I^{RGB}$, $I^{Bayer}$, and $I^{DTDI}$.

\subsubsection{Related work for $I^{RGB}$}
For the input RGB full-color image $I^{RGB}$, the five commonly used chroma subsampling methods are 4:2:0(A), 4:2:0(L), 4:2:0(R), 4:2:0(DIRECT), and 4:2:0(MPEG-B) \cite{MPEG-B}, and for convenience, the five traditional methods are denoted by the set symbol {\bf CS}. In addition, we introduce the two state-of-the-art chroma subsampling methods, namely the IDID (interpolation-dependent image downsampling) method \cite{Y.Zhang} and the MCIM (major color and index map-based) method \cite{S.Wang}.
For each 2$\times$2 UV block $B^{UV}$ in the chroma image $I^{UV}$, the purpose of each existing chroma subsampling method is to determine the subsampled $(U, V)$-pair of $B^{UV}$.

4:2:0(A) determines the subsampled $(U, V)$-pair of $B^{UV}$ by averaging the four U entries and the four $V$ entries of $B^{UV}$, separately. 4:2:0(L) and 4:2:0(R) determine the subsampled $(U, V)$-pairs by averaging the two chroma entries in the left column and right column of $B^{UV}$, respectively. 4:2:0(DIRECT) determines the subsampled $(U, V)$-pair by selecting the top-left $(U, V)$-entry of $B^{UV}$. 4:2:0(MPEG-B) determines the subsampled $(U, V)$-pair by performing the 13-tap filter with mask [2, 0, -4, -3, 5, 19, 26, 19, 5, -3, -4, 0, 2]/64 on the top-left location of $B^{UV}$.

Deploying the new edge-directed interpolation \cite{X.Li} based upsampling process into chroma subsampling, Zhang {\it et al.} \cite{Y.Zhang} proposed an IDID method, and IDID outperforms the traditional chroma subsampling methods. Considering the two characteristics, few colors and flat background, in the computer-generated subimages in screen content images \cite{Lu-2011}, Wang {\it et al.} {\cite{S.Wang}} proposed a MCIM method to improve the IDID method.

\subsubsection{Related work for $I^{Bayer}$}\label{i_bayer}
Instead of saving RGB triple-value for each pixel in $I^{RGB}$, it saves only one color value, R, G, or B, for each pixel in $I^{Bayer}$ for reducing hardware costs. As shown in Fig. \ref{pic:CFA}, there are four Bayer CFA patterns which have been widely used in modern digital color cameras.
To reduce the paper length, our discussion focuses on $I^{Bayer}$ with the Bayer CFA pattern in Fig. \ref{pic:CFA}(a), denoted by [$G_1, R_2, B_3, G_4$], but our discussion is also applicable to the other three Bayer CFA patterns in Fig. \ref{pic:CFA}(b)-(d).

From the YUV-to-RGB conversion in Eq. (\ref{eq:YUV2RGB}), Chen {\it et al.} \cite{Chen-2009} observed that the R value is dominated by the Y and V values, and the B value is dominated by the Y and U values. Therefore, the subsampled $(U, V )$-pair of each $B^{UV}$ is set to $(U_3, V_2)$.
Although their method benefits the quality of the R and B components in the reconstructed Bayer CFA image, it does not benefit the quality of the reconstructed G component.

To improve Chen {\it et al.}’s chroma subsampling method, at the server side, Lin {\it et al.} \cite{C.Lin} proposed a 2$\times$2 Bayer CFA block-distortion function under the COPY-based upsampling process where the four estimated $(U, V)$-pairs of $B^{UV}$ just copy the subsampled $(U, V)$-pair of $B^{UV}$.
Then, differentiating the Bayer CFA block-distortion function, a closed form was derived to determine the subsampled $(U, V)$-pair of $B^{UV}$.
For convenience, Lin {\it et al.}’s method is called the DI (differentiating) method.

To improve the DI method \cite{C.Lin}, Chung {\it et al.} \cite{Y.Lee} first proved that Lin {\it et al.}’s block-distortion function is a convex function, and then proposed a GD (gradient-descent) chroma subsampling method to achieve better quality of the reconstructed Bayer CFA images.
The common weakness and limitation in DI and GD is that the COPY-based estimation way used to estimate the four $(U, V)$-pairs of $B^{UV}$ is too simple, limiting the quality performance of the reconstructed Bayer CFA images.

%Utilizing the similar differentiation technique and the COPY-based chroma upsampling process \cite{C. Lin}, but deploying the converted 2$\times$2 RGB full-color block into the block-distortion, Lin {\it et al.} \cite{Lin-2019} proposed a modified 4:2:0(A) method by selecting the best case in the four cases on truncating and carrying. At the client side, they improved the previous chroma upsampling process \cite{Yu-2017} by considering the distance between each estimated upsampled chroma value and its three neighboring (TN) pixels; for convenience, their combination is called ``modified 4:2:0(A)-TN''.

\subsubsection{Related work for $I^{DTDI}$}\label{i_dtdi}
The DTDI CFA pattern has been commonly used in industrial high-speed line scan cameras. As depicted in Fig. \ref{pic:DTDI}, there are two 2$\times$2 DTDI patterns, where 
each pixel contains one G value and one R value (or B value).
Fig. \ref{pic:DTDI}(a) illustrates the 2$\times$2 DTDI CFA pattern [($G_1$, $B_1$), ($G_2$, $R_2$), ($G_3$, $B_3$), ($G_4$, $R_4$)]. Our introduction and discussion focus on the DTDI CFA pattern in Fig. \ref{pic:DTDI}(a), but our introduction and discussion are also applicable to Fig. \ref{pic:DTDI}(b).

From Eq. (\ref{eq:YUV2RGB}), Chung et al. \cite{W.Yang} observed that the B color is dominated by the Y and U components, and the R color is dominated by the Y and V components.
Therefore, for $I^{DTDI}$, the subsampled U component and the V component of $B^{UV}$ are determined by performing 4:2:0(L) and 4:2:0(R) on $B^{UV}$, respectively, leading to better quality of the reconstructed DTDI CFA images relative to the traditional chroma subsampling methods.
For convenience, their method is called the CD (color domination) method for $I^{DTDI}$.

\begin{figure}
  \centering
    \subfigure[]{\includegraphics[height=2cm]{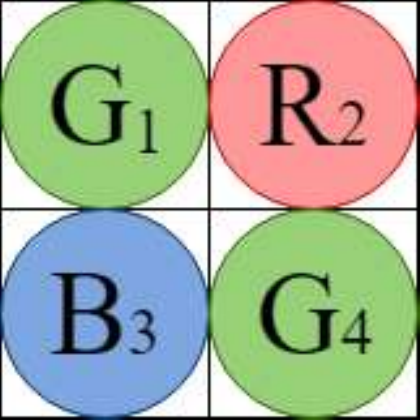}}
    \subfigure[]{\includegraphics[height=2cm]{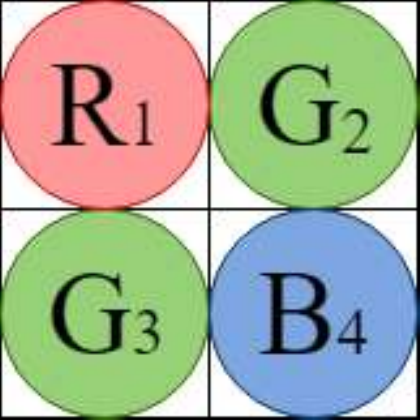}}
    \subfigure[]{\includegraphics[height=2cm]{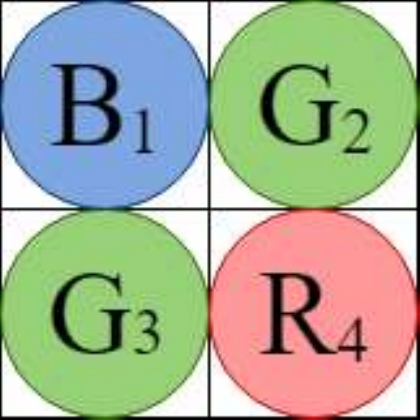}}
    \subfigure[]{\includegraphics[height=2cm]{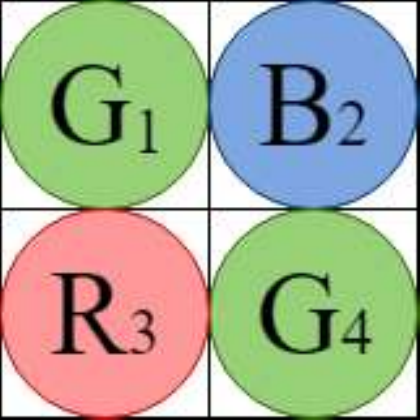}}
  \caption{Four Bayer CFA patterns. (a) [$G_1, R_2, B_3, G_4$]. (b) [$R_1, G_2, G_3, B_4$]. (c) [$B_1, G_2, G_3, R_4$]. (d) [$G_1, B_2, R_3, G_4$].}
  \label{pic:CFA}
\end{figure}

\begin{figure}
  \centering
    \subfigure[]{\includegraphics[height=2cm]{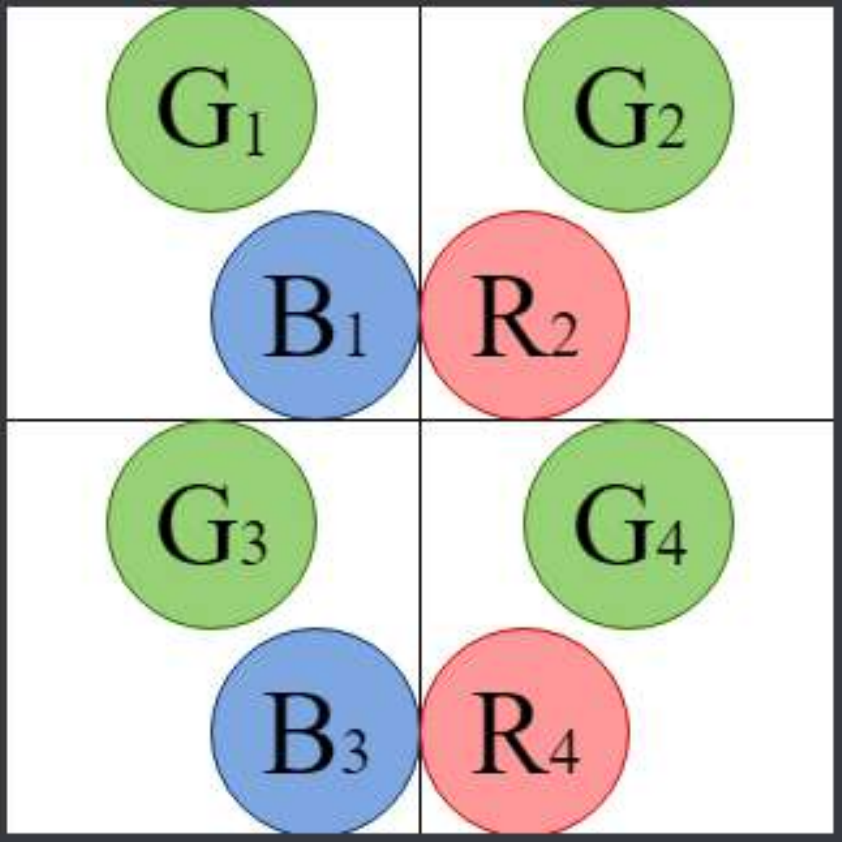}}
    \subfigure[]{\includegraphics[height=2cm]{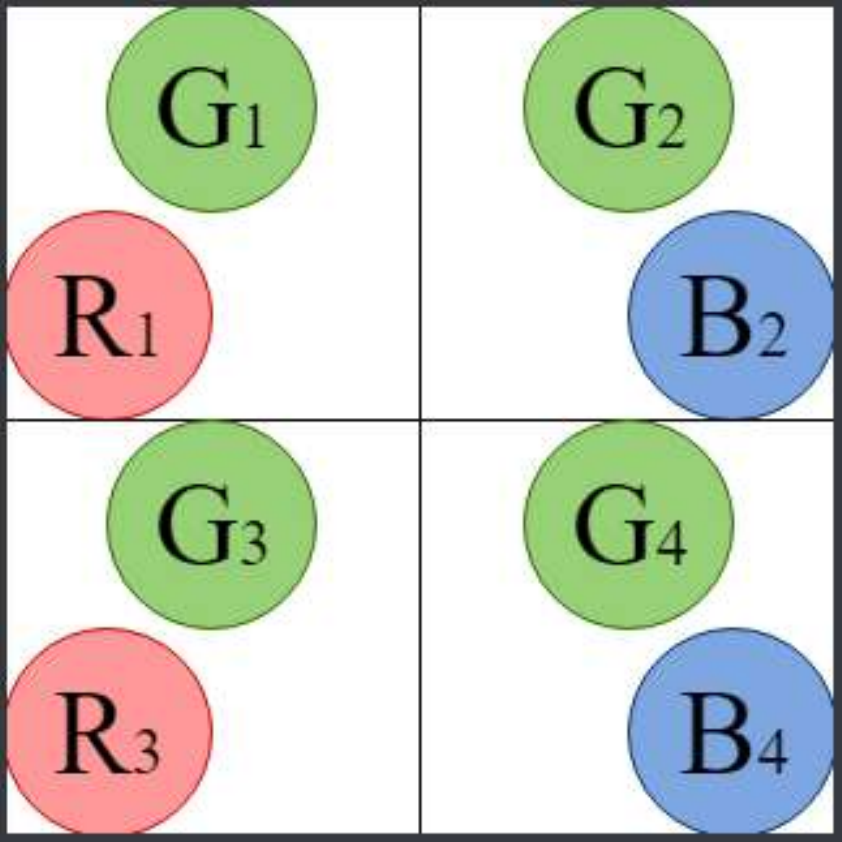}}
  \caption{Two DTDI CFA patterns. (a) [($G_1$, $B_1$), ($G_2$, $R_2$), ($G_3$, $B_3$), ($G_4$, $R_4$)]. (b) [($G_1$, $R_1$), ($G_2$, $B_2$), ($G_3$, $R_3$), ($G_4$, $B_4$)].}
  \label{pic:DTDI}
\end{figure}

\subsection{Motivation} \label{sec:IB}
%\begin{figure}
%  \centering
%    \subfigure[]{
%    \includegraphics[height=1.5 cm]{fig//GRBG.pdf}}
%    \subfigure[]{
%    \includegraphics[height=1.5 cm]{fig//RGGB.pdf}}
%    \subfigure[]{
%    \includegraphics[height=1.5 cm]{fig//BGGR.pdf}}
%    \subfigure[]{
%    \includegraphics[height=1.5 cm]{fig//GBRG.pdf}}
%  \caption{Four Bayer CFA modules. (a) $[$G, R, B, G$]$. (b) $[$R, G, G, B$]$. (c)$[$B, G, G, R$]$. (d) $[$G, B, R, G$]$.}
%  \label{pic:CFA}
%\end{figure}
%\begin{figure}
%  \centering
%    \subfigure[]{
%    \includegraphics[height=4 cm]{fig//4x4DTDI.pdf}}
%  \caption{A 4$\times$4 DTDI subimage with four 2$\times$2 DTDI blocks, each block with a DTDI CFA module $[(G_1, B_1), (G_2, R_2), (G_3, B_3), (G_4, R_4)]$.}
%  \label{pic:4x4DTDI}
%\end{figure}
From the related work introduction for $I^{RGB}$, $I^{Bayer}$, and $I^{DTDI}$, the first purpose of this paper is to propose an improved 2$\times$2 $t$ ($\in \{RGB, Bayer, DTDI\}$) color block-distortion model. The second purpose of this paper is to propose an improved chroma subsampling method for $I^{RGB}$, $I^{Bayer}$, and $I^{DTDI}$, achieving better quality and quality-bitrate tradeoff of the reconstructed RGB full-color, Bayer CFA, and DTDI CFA images relative to the existing traditional and state-of-the-art methods.

\subsection{Contributions} \label{sec:IC}
The three contributions of this paper are clarified as follows.

In the first contribution, we propose an improved 2$\times$2 $t$ ($\in \{RGB, Bayer, DTDI\}$) color block-distortion function at the server side. Furthermore, we prove that the proposed 2$\times$2 color block-distortion function is a convex function.

In the second contribution, for each 2$\times$2 UV block $B^{UV}$, we apply a convex function minimization technique to obtain the initially subsampled $(U, V)$-pair of $B^{UV}$.
Furthermore, based on the shape similarity between the convex function (corresponding to the proposed 2$\times$2 color block-distortion function) in real domain and the convex function in integer domain, we propose an iterative chroma subsampling method to better improve the subsampled $(U, V)$-pair of $B^{UV}$. 

In the third contribution, based on the Kodak dataset \cite{Kodak_dataset} and the IMAX dataset \cite{IMAX_dataset}, the comprehensive experimental results demonstrated that on the versatile video coding (VVC) platform VTM-8.0 \cite{VVC}, our improved chroma subsampling method achieves better quality, in terms of CPSNR (color peak signal-to-noise ratio), SSIM (structure similarity index) \cite{Z.Wang}, and FSIM (feature similarity index) \cite{FSIM}, and better quality-bitrate tradeoff relative to the five traditional methods in \textbf{CS}, IDID \cite{Y.Zhang}, and MCIM \cite{S.Wang} for $I^{RGB}$.
For $I^{Bayer}$, our improved method achieves better quality and quality-bitrate tradeoff relative to 4:2:0(A), the DI method \cite {C.Lin}, and the GD method \cite{Y.Lee}.
For $I^{DTDI}$, our improved method outperforms 4:2:0(A) and the CD method \cite{W.Yang}.

The rest of this paper is organized as follows. In Section \ref{sec:II}, our improved 2$\times$2 $t$ ($\in \{RGB, Bayer, DTDI\}$) color block-distortion model is presented, and then we prove that the proposed color block-distortion function is a convex function. In Section \ref{sec:III}, our improved chroma subsampling method is presented. In Section IV, the experimental results are demonstrated to justify the significant quality and quality-bitrate tradeoff merits of our improved method. In Section \ref{sec:VI}, some discussions and concluding remarks are addressed.

\section{OUR IMPROVED BLOCK-DISTORTION FUNCTION FOR 2$\times$2 RGB FULL-COLOR, BAYER CFA, AND DTDI CFA BLOCKS } \label{sec:II}
 In the first subsection, at the server side, we describe how to estimate the four $(U, V)$-pairs of each 2$\times$2 UV block $B^{UV}$, and then the estimated four $(U, V)$-pairs of $B^{U, V}$ are plugged in our improved 2$\times$2 $t$ ($\in \{RGB, Bayer, DTDI\}$) color block-distortion function.
 In the second subsection, our improved 2$\times$2 color block-distortion function is presented. In the third subsection, the convex property of our block-distortion function is proved.

\subsection{Estimating the Four $(U, V)$-pairs of Each 2$\times$2 UV Block } \label{sec:IIA}
\begin{figure}
  \centering
    \includegraphics[height = 5 cm]{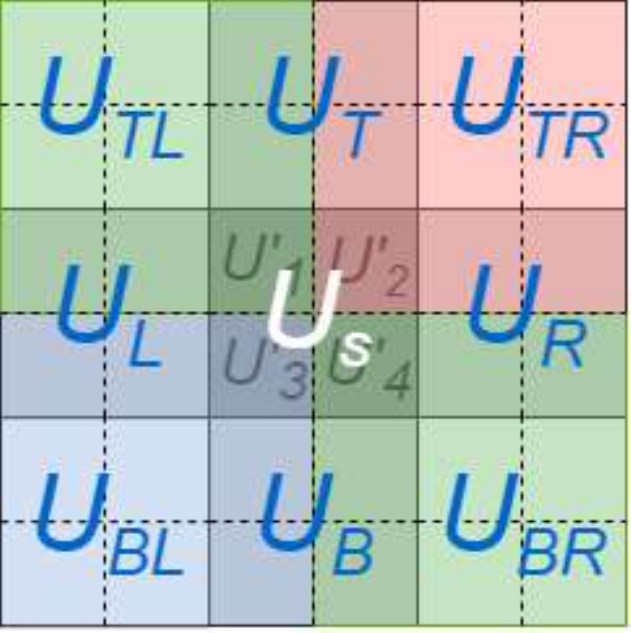}
  \caption{The notations used in the bilinear interpolation-based chroma upsampling at the server side.}
  \label{pic:3x3Bilinear}
\end{figure}
Let the input color image be denoted by $I^t$, where the symbol $``t "$ has been defined before.
For each 2$\times$2 UV block $B^{UV}$, let the subsampled $(U, V)$-pair of $B^{UV}$ be denoted by the parameter-pair $(U_s , V_s)$ which will be to be determined. 
Before defining our improved reconstructed 2$\times$2 color block-distortion function at the server side, we need to estimate the four $(U, V)$-pairs of $B^{UV}$, in which each estimated $(U, V)$-pair is expressed as a linear function with the parameter-pair $(U_s, V_s)$. 
To better estimate the four $(U, V)$-pairs of $B^{UV}$, besides the parameter-pair $(U_s, V_s)$, we also refer four known neighboring $(U, V )$-pairs, and four future neighboring $(U, V)$-pairs of $B^{UV}$.

For simplicity, we only describe how to estimate the four U entries of $B^{UV}$.
As depicted in Fig. \ref{pic:3x3Bilinear}, let the top-left, top-right, bottom-left, and bottom-right estimated U entries of $B^{UV}$ be denoted by $U'_1$, $U'_2$, $U'_3$, and $U'_4$, respectively; let the four known neighboring U entries be denoted by $U_{TL}$, $U_T$, $U_{TR}$, and $U_L$; let the four future neighboring U entries be denoted by $U_R$, $U_{BL}$, $U_B$, and $U_{BR}$.
Because our improved chroma subsampling method, which will be presented in Section \ref{sec:III}, operates in row-major order, the four known neighboring U entries have already obtained, while the four future neighboring U entries can be obtained by performing one traditional chroma subsampling method ($\in$ {\bf CS}), e.g. 4:2:0(A), on the four future 2$\times$2 UV blocks.

Applying the bilinear interpolation on the parameter-pair $(U_s , V_s)$ and the eight neighboring reference U entries, we have the following result.
\begin{prop}\label{pro1}

The four estimated U entries of $B^{U}$, namely $U'_1$, $U'_2$, $U'_3$, and $U'_4$, are expressed as

\end{prop}
\begin{equation}\label{eq:rec U}
  \begin{aligned}
    &U'_i=  \frac{9}{16}U_s + \bar{U}_i\\
  \end{aligned}
\end{equation}
for $1 \leq i \leq 4$ with 
\begin{equation}\label{eq:part U}
  \begin{aligned}
    &\bar{U}_{1} = \frac{1}{16}U_{TL}+\frac{3}{16}U_{T}+\frac{3}{16}U_{L}\\
    &\bar{U}_{2} = \frac{1}{16}U_{TR}+\frac{3}{16}U_{T}+\frac{3}{16}U_{R}\\
    &\bar{U}_{3} = \frac{1}{16}U_{BL}+\frac{3}{16}U_{B}+\frac{3}{16}U_{L}\\
    &\bar{U}_{4} = \frac{1}{16}U_{BR}+\frac{3}{16}U_{B}+\frac{3}{16}U_{R}\\
  \end{aligned}
\end{equation}

\begin{proof}\renewcommand{\qedsymbol}{}
We move the detailed proof to \ref{FirstAppendix} in order to make the paper more readable.
\end{proof}
Similar to the proof in Proposition \ref{pro1}, the four estimated V entries of $B^{UV}$ can be followed and they have the same expressions as that in Eqs. (\ref{eq:rec U})-(\ref{eq:part U}).

\subsection{Our Improved 2x2 Color Block-distortion Function for $I^{RGB}$, $I^{Bayer}$, and $I^{DTDI}$}\label{IIC}
As defined before, it is known that the notation ``$t$'' ($\in \{RGB, Bayer, DTDI\}$) denotes the considered image type.
For example, when $t$ = “$Bayer$”, $I^t$ (= $I^{Bayer}$) denotes the Bayer CFA image. For each 2$\times$2 $t$ color block in $I^t$, let $S^t_i$, $1\leq i \leq 4$, denote the color set in the $i$th pixel of the 2$\times$2 $t$ color block.
For example, when $t$ = “$RGB$”, we have ($S^{RGB}_1$, $S^{RGB}_2$, $S^{RGB}_3$, $S^{RGB}_4$) = (($R_1$, $G_1$, $B_1$), ($R_2$, $G_2$, $B_2$), ($R_3$, $G_3$, $B_3$), ($R_4$, $G_4$, $B_4$)) for each 2$\times$2 RGB full-color block $B^{RGB}$. When $t$ = ``$Bayer$'', as depicted in Fig. \ref{pic:CFA}(a), we have $(S^{Bayer}_1, S^{Bayer}_2, S^{Bayer}_3, S^{Bayer}_4) = (G_1, R_2, B_3, G_4)$ for each 2$\times$2 Bayer CFA block $B^{Bayer}$. When $t$ = ``$DTDI$'', as depicted in Fig. \ref{pic:DTDI}(a), we have ($S^{DTDI}_1$, $S^{DTDI}_2$, $S^{DTDI}_3$, $S^{DTDI}_4$) = (($G_1$, $B_1$), ($G_2$, $R_2$), ($G_3$, $B_3$), ($G_4$, $R_4$)) for each 2$\times$2 DTDI CFA block $B^{DTDI}$.

By Proposition \ref{pro1}, the four estimated $(U, V)$-pairs of $B^{UV}$, namely $(U'_i, V'_i)$ for $1\leq i\leq 4$, can be obtained. 
Furthermore, for $1\leq i\leq 4$, in Eq. (\ref{eq:YUV2RGB}), we replace $U_i$ and $V_i$ by $U'_i$ and $V'_i$, respectively, and then the four estimated $t$ color values, $t \in \{RGB, Bayer, DTDI\}$, can be obtained at the server side.
For convenience, let the estimated 2$\times$2 $t$ color block be denoted by $B^{est,t}$ and the 2$\times$2 ground truth $t$ color block be denoted by $B^{t}$ which comes from $I^{t}$.
In what follows, as the 2$\times$2 $t$ color block-distortion, we explain how to express the sum of squared errors between $B^{est,t}$ and $B^t$ in more detail.

We first consider $t$ = “$RGB"$, and then the 2$\times$2 RGB full-color block-distortion between $B^{est,RGB}$ and $B^{RGB}$ is expressed as $\sum_{i=1}^4 [(R_i - R'_i)^2 + (G_i - G'_i)^2 + (B_i - B'_i)^2]$.
For $t$ = “$Bayer$", the 2$\times$2 Bayer CFA block-distortion is expressed as $D^{Bayer}$ = $[(G_1 - G'_1)^2 + (R_2 - R'_2)^2 + (B_3 - B'_3)^2 + (G_4 - G'_4)^2]$. 
Similarly, when $t$ = “$DTDI$", the 2$\times$2 DTDI CFA block-distortion is expressed as 
$D^{DTDI} = [(G_1 - G'_1)^2 + (B_1 - B'_1)^2 + (G_2 - G'_2)^2 + (R_2 - R'_2)^2 + (G_3 - G'_3)^2 + (B_3 - B'_3)^2 + (G_4 - G'_4)^2 + (R_4 - R'_4)^2]$.

For easy exposition, we consider $t$ = “$Bayer$”. By Eq. (\ref{eq:RGB2YUV}), the 2$\times$2 Bayer CFA block-distortion function $D^{Bayer}(U_s, V_s)$ can be expressed as a function of the two parameters $U_s$ and $V_s$ by the following derivation:
\begin{equation}\label{eq:General BY}
\begin{small}
\begin{aligned}
    &D^{Bayer}(U_s, V_s) \\
    &=[(G_1 - G'_1)^2 + (R_2 - R'_2)^2 + (B_3 - B'_3)^2 + (G_4 - G'_4)^2] \\
    &=[(- 0.391(U_1 - U'_1)- 0.813(V_1 - V'_1))^2 \\
    &+ (1.596(V_2 - V'_2))^2 + (2.018(U_3 - U'_3))^2\\
    &+ (- 0.391(U_4 - U'_4)- 0.813(V_4 - V'_4))^2] \\
    &=\sum\limits_{i=1}\limits^{4}\sum\limits_{c \in S^{Bayer}_i}[a_c(U'_i-U_i)+b_c(V'_i-V_i)]^2\\
    &=\sum\limits_{i=1}\limits^{4}\sum\limits_{c \in S^{Bayer}_i}[a_c(\frac{9}{16}U_s+\bar{U}_i-U_i)+b_c(\frac{9}{16}V_s+\bar{V}_i-V_i)]^2
\end{aligned}
\end{small}
\end{equation}
where the values of $a_c$ and $b_c$ are defined by
\begin{equation}
  \label{eq:ab}
  \begin{aligned}
  &a_c=
  \begin{cases}
      0       &\text{for } c = R_i\\
      -0.391  &\text{for } c = G_i\\
      2.018   &\text{for } c = B_i\\
  \end{cases}\\
  &\\
  &b_c=
  \begin{cases}
      1.596   &\text{for } c = R_i\\
      -0.813  &\text{for } c = G_i\\
      0       &\text{for } c = B_i\\
  \end{cases}
\end{aligned}
\end{equation}
In the same way, in terms of $U_s$ and $V_s$, the 2$\times$2 RGB full-color and DTDI CFA block-distortion functions, $D^{RGB}(U_s, V_s)$ and $D^{DTDI}(U_s, V_s)$, can be followed.
In general, our improved 2$\times$2 $t$ ($\in \{RGB, Bayer, DTDI\}$) color block-distortion is expressed as
\begin{equation}\label{eq:General BD}
  \begin{aligned}
     &D^{t}(U_s, V_s) 
     =\sum\limits_{i=1}\limits^{4}\sum\limits_{c \in S^{t}_i}[a_c(U'_i-U_i)+b_c(V'_i-V_i)]^2\\
      &=\sum\limits_{i=1}\limits^{4}\sum\limits_{c \in S^{t}_i}[a_c(\frac{9}{16}U_s+\bar{U}_i-U_i)+b_c(\frac{9}{16}V_s+\bar{V}_i-V_i)]^2
  \end{aligned}
\end{equation}

\subsection{Convex Property Proof for the 2$\times$2 $t$ ($\in \{RGB, Bayer, DTDI\}$) Color Block-distortion Function} \label{sec:IIB}

We now prove that our improved 2$\times$2 $t$ ($\in \{RGB, Bayer, DTDI\}$) color block-distortion function in Eq. (\ref{eq:General BD}) is a convex function, and we have the following result.

\begin{prop}\label{thm:convex}
Our improved 2$\times$2 $t$ ($\in \{RGB, Bayer, DTDI\}$) color block-distortion function in Eq. (\ref{eq:General BD}) is a convex function. 
\end{prop}
\begin{proof}\renewcommand{\qedsymbol}{}
To make the paper more readable, we move the detailed proof to \ref{proof}.
\end{proof}

%\begin{figure}[h]
%  \centering
%    \subfigure[]{
%    \includegraphics[height=3 cm]{fig//example2x2RGB.pdf}}
%      \hspace{0.4 in}
%    \subfigure[]{
%    \includegraphics[height=3 cm]{fig//example2x2YUV.pdf}}\\
      %\hspace{0.05 in}
%    \subfigure[]{
%    \includegraphics[height=3.5 cm]{fig//example6x6UV.pdf}}
%      %\hspace{0.05 in}
%    \subfigure[]{
%    \includegraphics[height=3.5 cm]{fig//real.pdf}}
%  \caption{One 2$\times$2 RGB full-color block example and the plot of its convex block-distortion function in the real domain. (a) One 2$\times$2 RGB full-color block example. (b) The transformed YUV block. (c) The eight neighboring known $(U, V)$-pairs referred by the current UV block $B_c$. (d) The convex shape of the plot for the block-distortion function.}
%  \label{pic:RGB convex example}
%\end{figure}

\section{OUR IMPROVED CHROMA SUBSAMPLING METHOD FOR RGB FULL-COLOR, BAYER CFA, AND DTDI CFA IMAGES} \label{sec:III}

 In the first subsection, we apply the differentiation technique on our improved $t$ ($\in \{RGB, Bayer, DTDI\}$) color block-distortion function, which is a convex function as proved in Proposition \ref{thm:convex}, to obtain the closed form which will be used as the initially subsampled $(U, V)$-pair of each 2$\times$2 UV block $B^{UV}$.
 In the second subsection, according to the shape similarity between the convex function in real domain and that in integer domain, we propose an iterative procedure to better improve the initially subsampled $(U, V)$-pair of $B^{UV}$ in the integer domain [0, 255] $\times$ [0, 255].
 
\subsection{Deriving the Closed Form as the the Initially Subsampled $(U, V)$-pair of $B^{UV}$} \label{sec:IIIA}

By taking the first derivative on Eq. (\ref{eq:General BD}) with respect to $U_s$ and $V_s$, respectively, and then setting the two derivatives to zero, it yields
\begin{equation}
  \label{eq:General derivatives}
  \begin{aligned}
    \frac{\partial D^t(U_s, V_s)}{\partial {V_s}} = 0\\
    \frac{\partial D^t(U_s, V_s)}{\partial {U_s}} = 0.
  \end{aligned}
\end{equation}
After solving the two equations in Eq. (\ref{eq:General derivatives}), the solution is denoted by $(U^{(0),t}_{s}, V^{(0),t}_{s})$ which is expressed as the closed form in Eq. (\ref{eq:General solutions}), and the solution is exactly the critical point of our improved $t$ ($\in \{RGB, Bayer, DTDI\}$) color block-distortion function in Eq. (\ref{eq:General BD}).
Naturally, the closed form in Eq. (\ref{eq:General solutions}) can be used as the initial chroma subsampling solution of $B^{UV}$ since it can minimize our improved $t$ color block-distortion function in real domain.

\begin{figure*}[ht]
  \begin{equation}
    \begin{aligned}\\
      \label{eq:General solutions}
      &U^{(0),t}_{s} = \frac{(\sum\limits_{i=1}\limits^{4}\sum\limits_{c \in S^t_i}b_{c}^{2})
      \cdot[\sum\limits_{i=1}\limits^{4}\sum\limits_{c \in S^t_i}a_c^2(\bar{U}_i-U_i)+a_cb_c(\bar{V}_i-V_i)]
      -(\sum\limits_{i=1}\limits^{4}\sum\limits_{c \in S^t_i}a_cb_c)
      \cdot[\sum\limits_{i=1}\limits^{4}\sum\limits_{c \in S^t_i}b_c^2(\bar{V}_i-V_i)+a_cb_c(\bar{U}_i-U_i)]}
      {\frac{9}{16}[(\sum\limits_{i=1}\limits^{4}\sum\limits_{c \in S^t_i}a_cb_c)^2
      -(\sum\limits_{i=1}\limits^{4}\sum\limits_{c \in S^t_i}a_c^2)\cdot (\sum\limits_{i=1}\limits^{4}\sum\limits_{c \in S^t_i}b_c^2)]}\\
      &V^{(0),t}_{s} = \frac{(\sum\limits_{i=1}\limits^{4}\sum\limits_{c \in S^t_i}a_c^2)
      \cdot[\sum\limits_{i=1}\limits^{4}\sum\limits_{c \in S^t_i}b_c^2(\bar{V}_i-V_i)+a_cb_c(\bar{U}_i-U_i)]
      -(\sum\limits_{i=1}\limits^{4}\sum\limits_{c \in S^t_i}a_cb_c)
      \cdot[\sum\limits_{i=1}\limits^{4}\sum\limits_{c \in S^t_i}a_c^2(\bar{U}_i-U_i)+a_cb_c(\bar{V}_i-V_i)]}
      {\frac{9}{16}[(\sum\limits_{i=1}\limits^{4}\sum\limits_{c \in S^t_i}a_cb_c)^2
      -(\sum\limits_{i=1}\limits^{4}\sum\limits_{c \in S^t_i}a_c^2)\cdot (\sum\limits_{i=1}\limits^{4}\sum\limits_{c \in S^t_i}b_c^2)]}\\
  \end{aligned}
  \end{equation}
\end{figure*}

\subsection{The Proposed Iterative Chroma Subsampling Method for $t$ ($\in \{RGB, Bayer, DTDI\}$) Color Images}\label{IIIB}
Before presenting our iterative chroma subsampling method for $I^{RGB}$, $I^{Bayer}$, and $I^{DTDI}$, we analyze the shape similarity between the convex block-distortion function in real domain and that in integer domain in which the subsampled $(U, V)$-pair is practically considered in the integer domain [0, 255] $\times$ [0, 255].

\subsubsection{The Shape Similarity between the Convex Block-distortion Function in Real Domain and That in Integer Domain}\label{IIIB1}
We first take a real 2$\times$2 RGB full-color block $B^{RGB}$, as shown in Fig. \ref{pic:RGB convex example}(a), to describe the shape similarity between the convex RGB full-color block-distortion function (see Eq. (\ref{eq:General BD})) in real domain and that in the integer domain [0, 255] $\times$ [0, 255].
Then, we explain why the initially subsampled $(U, V)$-pair of the converted UV block $B^{UV}$, which has been obtained by Eq. (\ref{eq:General solutions}), has room to be better improved. 

By Eq. (\ref{eq:RGB2YUV}), we transform $B^{RGB}$, as shown in Fig. \ref{pic:RGB convex example}(a), to a 2$\times$2 YUV block $B^{YUV}$ which is shown in Fig. \ref{pic:RGB convex example}(b).
The eight neighboring $(U, V)$-pairs referred by the current 2$\times$2 UV block $B^{UV}$ are shown in Fig. \ref{pic:RGB convex example}(c).
In the real domain, Fig. \ref{pic:RGB convex example}(d) depicts the convex shape of the plot for the 2$\times$2 block-distortion function of Fig. \ref{pic:RGB convex example}(a).

\begin{table*}[]
\centering
\caption{THE CPSNR, SSIM, FSIM, AND TIME COMPARISON AMONG THE CONSIDERED METHODS FOR $I^{RGB}$.}
\label{Tab:RGB Table}
\begin{tabular}{|c|c|c|c|c|c|c|c|}
\hline
$I^{RGB}$                        & 4:2:0(A)   & 4:2:0(L)   & 4:2:0(R)        & 4:2:0(DIRECT)   & 4:2:0(MPEG-B) & IDID \cite{Y.Zhang} & \textbf{Our method}                \\ \hline
\multirow{2}{*}{CPSNR (dB)}      & 41.8772    & 42.6655    & 41.8246         & 43.0554         & 42.6817       & 42.9531             & \multirow{2}{*}{\textbf{43.8573}}  \\ \cline{2-7}
                                 & [43.1160]  & [43.2487]  & [42.6493]       & [42.9399]       & [42.9230]     & [42.6836]           &                                    \\ \hline
\multirow{2}{*}{CPSNR Gain (dB)} & 1.9800     & 1.1918     & 2.0327          & 0.8019          & 1.1756        & 0.9042              & \multirow{2}{*}{\textbf{}}         \\ \cline{2-7}
                                 & [0.7413]   & [0.6085]   & [1.2080]        & [0.9174]        & [0.9343]      & [1.1737]            &                                    \\ \hline
\multirow{2}{*}{SSIM}            & 0.9793     & 0.9812     & 0.9803          & 0.9817          & 0.9803        & 0.9815              & \multirow{2}{*}{\textbf{0.9854}}   \\ \cline{2-7}
                                 & [0.9831]   & [0.9830]   & [0.9824]        & [0.9820]        & [0.9814]      & [0.9815]            &                                    \\ \hline
\multirow{2}{*}{SSIM Gain}       & 0.0061     & 0.0043     & 0.0051          & 0.0037          & 0.0051        & 0.0040              & \multirow{2}{*}{\textbf{}}         \\ \cline{2-7}
                                 & [0.0024]   & [0.0025]   & [0.0031]        & [0.0035]        & [0.0040]      & [0.0040]            &                                    \\ \hline
\multirow{2}{*}{FSIM}            & 0.999742   & 0.999745   & 0.999727        & 0.999724        & 0.999723      & 0.999727            & \multirow{2}{*}{\textbf{0.999857}} \\ \cline{2-7}
                                 & [0.999839] & [0.999790] & [0.999780]      & [0.999728]      & [0.999743]    & [0.999705]          &                                    \\ \hline
\multirow{2}{*}{FSIM Gain}       & 0.000115   & 0.000112   & 0.000130        & 0.000133        & 0.000135      & 0.000130            & \multirow{2}{*}{\textbf{}}         \\ \cline{2-7}
                                 & [0.000018] & [0.000067] & [0.000077]      & [0.000129]      & [0.000114]    & [0.000152]          &                                    \\ \hline
Time (s)                         & 0.0015     & 0.0144     & \textbf{0.0008} & \textbf{0.0008} & 0.0027        & 6.1887              & 0.0355                             \\ \hline
\end{tabular}
\end{table*}

\begin{table}[]
\caption{THE CPSNR, SSIM, FSIM, AND TIME COMPARISON BETWEEN THE MCIM METHOD \cite{S.Wang} AND OUR METHOD FOR $I^{RGB}$.}
\centering
\label{Tab:Kw Table}
\begin{tabular}{|c|c|c|}
\hline
$I^{RGB}$                        & MCIM \cite{S.Wang}        & \textbf{Our method}                \\ \hline
\multirow{2}{*}{CPSNR (dB)}      & 43.2527     & \multirow{2}{*}{\textbf{43.8573}}  \\ \cline{2-2}
                                 & [43.5009]   &                                    \\ \hline
\multirow{2}{*}{CPSNR Gain (dB)} & 0.6046      & \multirow{2}{*}{\textbf{}}         \\ \cline{2-2}
                                 & [0.3564]    &                                    \\ \hline
\multirow{2}{*}{SSIM}            & 0.9821      & \multirow{2}{*}{\textbf{0.9854}}   \\ \cline{2-2}
                                 & [0.9832]    &                                    \\ \hline
\multirow{2}{*}{SSIM Gain}       & 0.0034      & \multirow{2}{*}{\textbf{}}         \\ \cline{2-2}
                                 & [0.0023]    &                                    \\ \hline
\multirow{2}{*}{FSIM}            & 0.999751    & \multirow{2}{*}{\textbf{0.999857}} \\ \cline{2-2}
                                 & [0.999763 ] &                                    \\ \hline
\multirow{2}{*}{FSIM Gain}       & 0.000106    & \multirow{2}{*}{\textbf{}}         \\ \cline{2-2}
                                 & [0.000094]  &                                    \\ \hline
Time (s)                         & 0.0787      & \textbf{0.0355}                    \\ \hline
\end{tabular}
\end{table}

\begin{table*}[]
\centering
\caption{THE PSNR, SSIM, FSIM, AND TIME COMPARISON AMONG THE CONSIDERED METHODS FOR $I^{Bayer}$.}
\label{Tab:CFA Table}
\scalebox{1}{
\begin{tabular}{|c|c|c|c|c|}
\hline
$I^{Bayer}$ & 4:2:0(A)-BILI & DI-COPY \cite{C.Lin} & GD-BILI \cite{Y.Lee} & Our combination  \\ \hline
PSNR(dB)    & 42.0728       & 45.4834                              & 47.1694                 & \textbf{48.4807} \\ \hline
PSNR Gain   & 6.4080        & 2.9973                               & 1.3114                  & \textbf{}        \\ \hline
SSIM(dB)    & 0.9961        & 0.9981                               & 0.9987                  & \textbf{0.9989}  \\ \hline
SSIM Gain   & 0.0028        & 0.0009                               & 0.0003                  & \textbf{}        \\ \hline
FSIM(dB)    & 0.99851       & 0.99833                              & 0.99922                 & \textbf{0.99930} \\ \hline
FSIM Gain   & 0.00079       & 0.00097                              & 0.00007                 & \textbf{}        \\ \hline
Time (s)     & 0.0013        & \textbf{0.0012}                      & 0.0255                  & 0.0214           \\ \hline
\end{tabular}}
\end{table*}

\begin{figure}
  \centering
    \subfigure[]{
    \includegraphics[height=3 cm]{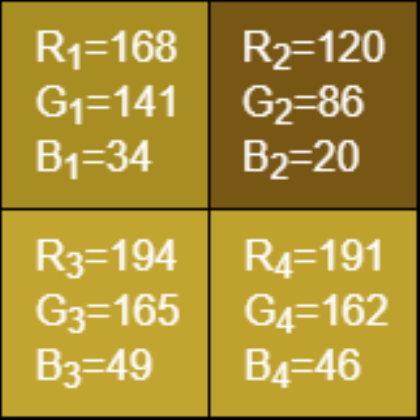}}
      \hspace{0.4 in}
    \subfigure[]{
    \includegraphics[height=3 cm]{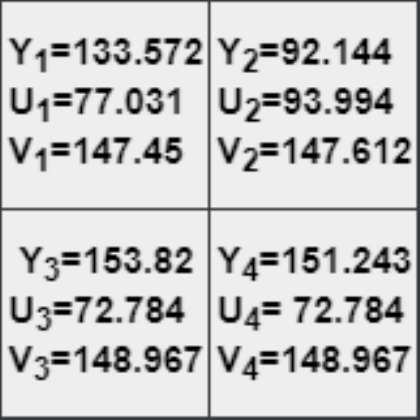}}\\
      %\hspace{0.05 in}
    \subfigure[]{
    \includegraphics[height=3.5 cm]{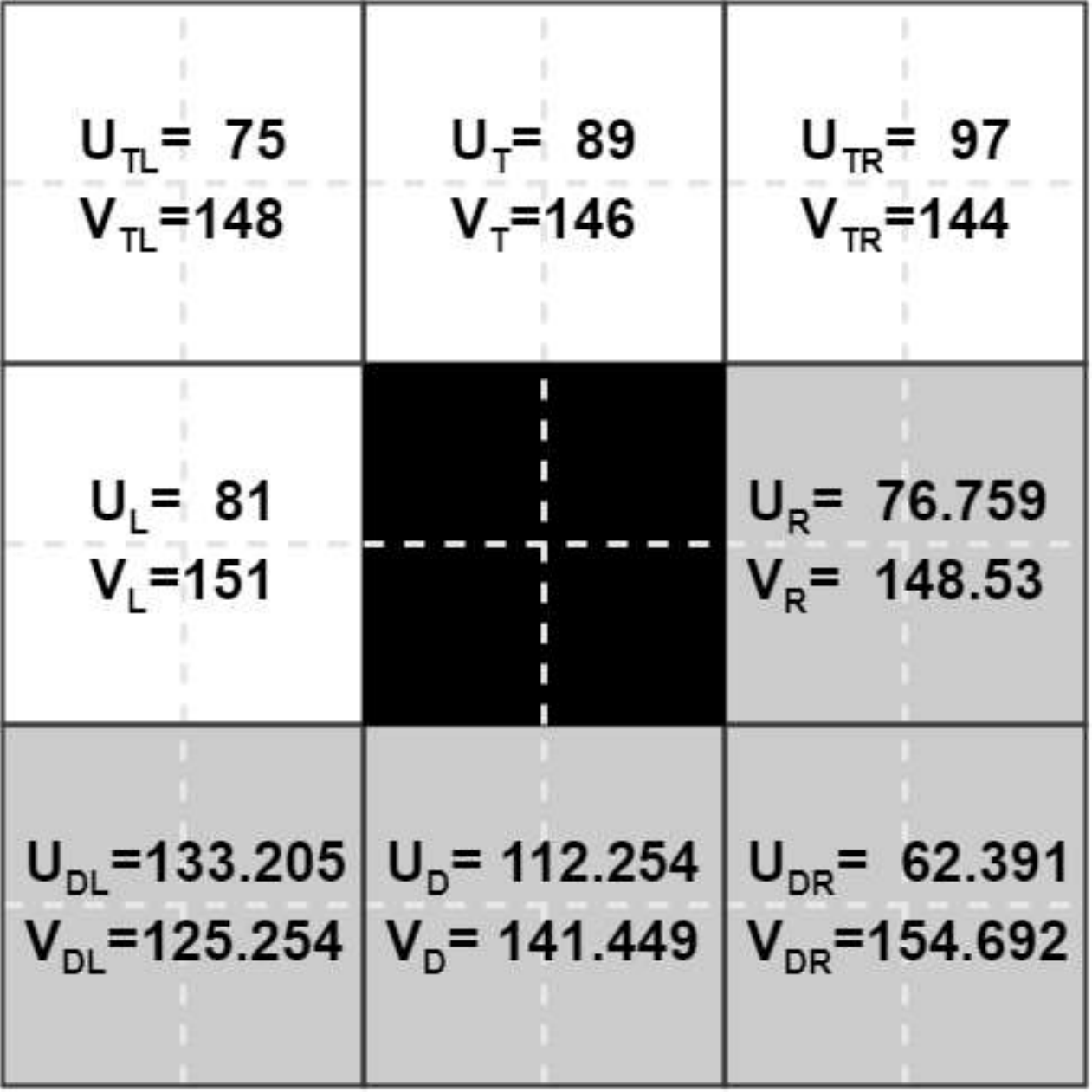}}
      %\hspace{0.05 in}
    \subfigure[]{
    \includegraphics[height=3.5 cm]{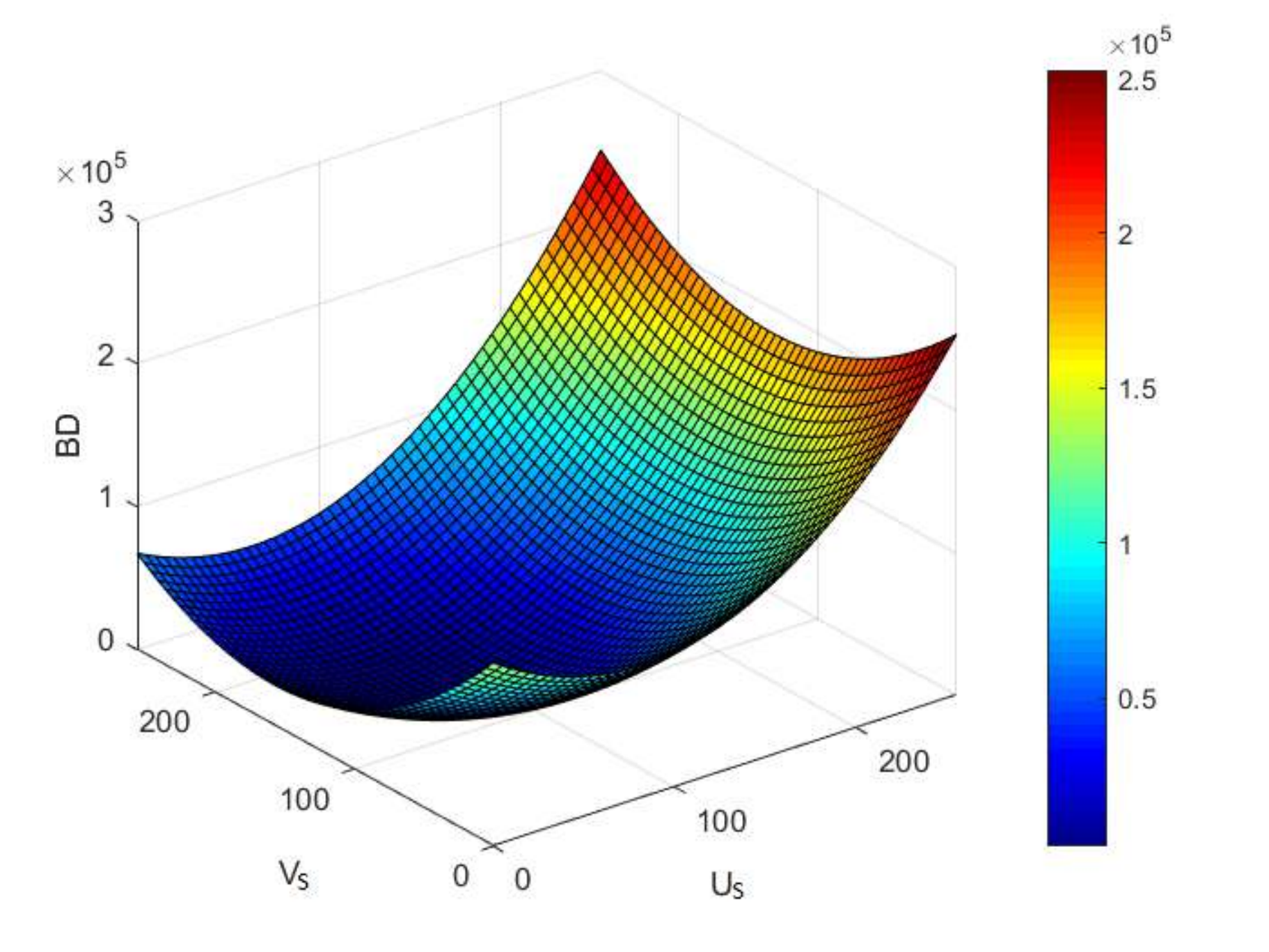}}
  \caption{One 2$\times$2 RGB full-color block example and the plot of its convex block-distortion function in the real domain. (a) The input 2$\times$2 RGB full-color block $B^{RGB}$. (b) The converted YUV block $B^{YUV}$. (c) The eight neighboring subsampled $(U, V)$-pairs referred by the current 2$\times$2 UV block $B^{UV}$. (d) The convex shape of the plot for the 2$\times$2 block-distortion function of Fig. \ref{pic:RGB convex example}(a).}
  \label{pic:RGB convex example}
\end{figure}
%\begin{figure}[h]
%  \centering
%    \subfigure[]{ \includegraphics[height=6 cm]{fig//exampleJCSU2IJCSU.pdf}}
%  \caption{The convex-like grid plot of Fig. \ref{pic:RGB convex example}(d) under the integer interval [0, 255] and the room to improve the initial chroma subsampling solution, $(U_s^{(0),t}, V_s^{(0),t})$.}
%  \label{pic:GD example}
%\end{figure}
We now discuss the shape similarity between the convex block-distortion function in real domain and that in the integer domain [0, 255] $\times$ [0, 255].
Under the integer domain [0, 255] $\times$ [0, 255], the discretely convex-like grid plot of Fig. \ref{pic:RGB convex example}(d) is depicted in Fig. \ref{pic:GD example}.
In Fig. \ref{pic:GD example}, the room to better improve the initial chroma subsampling solution $(U_s^{(0),RGB}, V_s^{(0),RGB})$ in which the superscript ``0'' denotes the initial step, is indicated by the path from the red point $(U_s^{(0),RGB}, V_s^{(0),RGB})$ to the yellow point $(U_s^{(k),RGB}, V_s^{(k),RGB})$ in which the superscript ``$k$'' denotes the $k$th step of our iterative chroma subsampling method which will be presented in Subsection \ref{IIIB2}.
Our iterative chroma subsampling method can result in less 2$\times$2 RGB full-color block-distortion value, and can achieve better reconstructed RGB full-color images.

In fact, besides the 2$\times$2 RGB full-color block $B^{RGB}$, the shape similarity between the convex block-distortion function in real domain and that in integer domain [0, 255] $\times$ [0, 255] and the room to better improve the initial chroma subsampling solution are also feasible to the 2$\times$2 Bayer CFA block $B^{Bayer}$ and the 2$\times$2 DTDI CFA block $B^{DTDI}$. 
\begin{figure}[h]
  \centering
    \subfigure[]{ \includegraphics[height=6 cm]{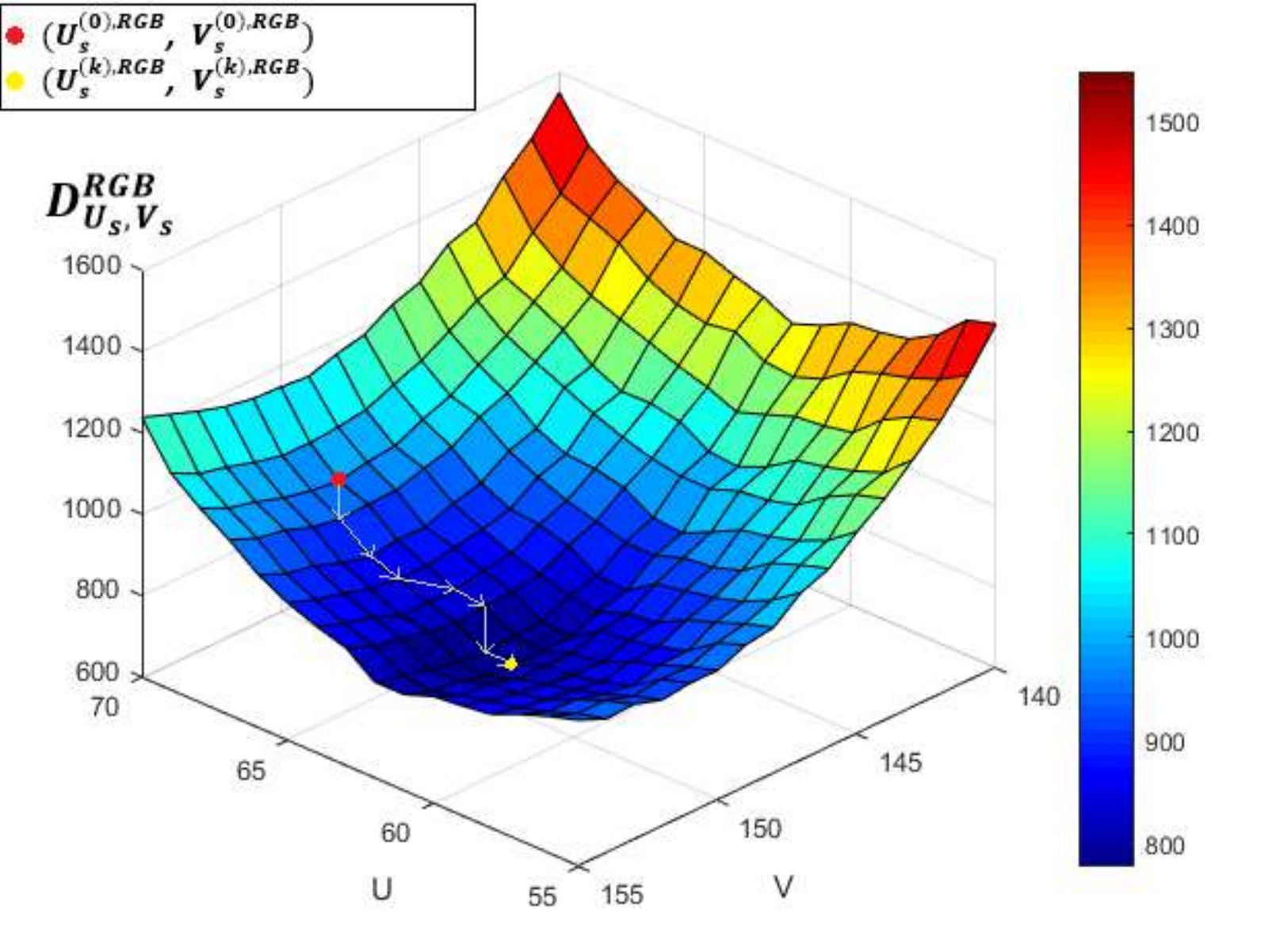}}
  \caption{The convex-like grid plot of Fig. \ref{pic:RGB convex example}(d) under the integer domain [0, 255] $\times$ [0, 255] and the room to better improve the initial chroma subsampling solution.}
  \label{pic:GD example}
\end{figure}

\subsubsection{Our Iterative Chroma Subsampling Procedure}\label{IIIB2}

As shown in Fig. \ref{pic:GD example}, from the room to better improve the initial chroma subsampling solution (see the red point in Fig. \ref{pic:GD example}), we now present our iterative chroma subsampling procedure to reduce the 2$\times$2 $t$ ($\in \{RGB, Bayer, DTDI\}$) color block-distortion value for each 2$\times$2 $t$ color block $B^t$,
resulting in better quality of the reconstructed $t$ color images.
The numerical technique used in our iterative chroma subsampling procedure for the three kinds of color images is quite similar to the gradient-descent technique.
For the input $t$ ($\in \{RGB, Bayer, DTDI\}$) color block $B^t$, taking $(U_s^{(0),t}, V_s^{(0),t})$ as the initial chroma subsampling solution, our iterative chroma subsampling procedure is shown below.

\begin{algorithm}[h]
  \small

  \caption{Our Improved Iterative Chroma Subsampling Method}
 
  \label{algo:IJCSU}
  \KwIn{2$\times$2 $t$ color block $B^t$.}
  \KwOut{Subsampled $(U, V)$-pair of $B^{UV}$, $(U_s^{(k),t}, V_s^{(k),t})$.}
  {\bf Step 1:} By Eq. (\ref{eq:General solutions}), we take $(U_s^{(0),t}$, $V_s^{(0),t})$ as the initial chroma subsampling solution of the current 2$\times$2 UV block $B^{UV}$; by Eq. (\ref{eq:General BD}), we calculate the $t$ color block-distortion, denoted by $D^t(U_s^{(0),t}, V_s^{(0),t})$. Set $k = 0$.\\
  {\bf Step 2:} Under the integer domain [0, 255] $\times$ [0, 255], we calculate all the eight neighboring $t$ color block-distortion values of $(U_s^{(k),t}, V_s^{(k),t})$, namely $D^t(U_s^{(k),t} + m, V_s^{(k),t} + n)$ for $(m, n)$ $\in \{(0, 1), (0, -1), (1, 0), (-1, 0), (1, 1), (1, -1), (-1, 1),$ $(-1, -1)\}$.
  Among the eight block-distortion values, we select the subsampled $(U, V)$-pair with the minimal block-distortion value as the candidate subsampled $(U, V)$-pair of $B^{UV}$, namely $(U_s^{(k+1),t}, V_s^{(k+1),t})$.\\
  {\bf Step 3:} If $D^t(U_s^{(k+1),t}, V_s^{(k+1),t})$ $\geq$ $D^t(U_s^{(k),t}, V_s^{(k),t})$, we stop the procedure and report $(U_s^{(k),t}, V_s^{(k),t})$ as the finally subsampled $(U, V)$-pair; otherwise, we perform $k := k+1$ and go to Step 2.
%\noindent\rule{80mm}{0.7pt}

\end{algorithm}

The execution codes of our Improved chroma subsampling method for the input $t$ ($\in \{RGB, Bayer, DTDI\}$) color image can be accessed from the website \cite{ICSCI}.

\section{Experimental Results} \label{sec:V}
Under the VVC reference software platform VTM-8.0, based on the Kodak and IMAX datasets, the thorough experimental results demonstrated the quality and quality-bitrate tradeoff merits of our improved chroma subsampling  method, in which the execution codes can be accessed from the website in \cite{ICSCI}, for $I^{RGB}$, $I^{Bayer}$, and $I^{DTDI}$ relative to the comparative methods. A comparison of the execution time of the considered methods is also made.

All the considered methods are implemented on a computer with an Intel Core i7-8700 CPU 3.2 GHz and 24 GB RAM. The operating system is the Microsoft Windows 10 64-bit operating system. The program development environment is Visual C++ 2019.

\subsection{Quality Merit and Time Comparison} \label{sec:VA}
When setting QP to zero, the PSNR, CPSNR, SSIM \cite{Z.Wang} and FSIM \cite{FSIM} metrics are used to show the quality merit of the reconstructed images using our improved chroma subsampling method relative to the existing methods. In these experiments, the related results are computed by passing the compression process and decompression process.\\

\subsubsection{For $I^{RGB}$ }
The CPSNR metric of the reconstructed RGB full-color image is defined by

\begin{equation}
  \label{eq:CPSNR}
  \text{CPSNR}=\frac{1}{N}\sum_{n=1}^{N}10\log_{10}\frac{255^2}{CMSE}
\end{equation}
with $CMSE=\frac{1}{3WH}\sum_{p\in P}\sum_{C\in\{R,G,B\}}[I_{n,C}(p)-I_{n,C}^{rec}(p)]^2$ in which the test image is of size WxH.
$I_{n,C}(p)$ and $I_{n,C}^{rec}(p)$ denote the C-color pixel-values at position $p$ in the $n$th input RGB full-color image and the reconstructed analogue, respectively.
$N$ denotes the number of testing images in the dataset.
Here, $N$ is equal to 24 and 18 for the Kodak dataset and the IMAX dataset, respectively.
We first calculate the CPSNR value of each dataset, and then calculate the average CPSNR value of the three related CPSNR values.

Based on the two datasets for $I^{RGB}$, Table \ref{Tab:RGB Table} indicates that our method has the highest CPSNR in boldface among the considered seven methods.
The average CPSNR gains of our method are 1.9800 dB, 1.1918 dB, 2.0327 dB, 0.8019 dB, 1.1756 dB, and 0.9042 dB over 4:2:0(A), 4:2:0(L), 4:2:0(R), 4:2:0(DIRECT), 4:2:0(MPEG-B), and the IDID method \cite{Y.Zhang}, respectively, using the bilinear interpolation-based chroma upsampling process at the client side.
Furthermore, using the bicubic interpolation-based chroma upsampling process at the client side, the average CPSNR values of the considered methods are listed in parentheses in Table \ref{Tab:RGB Table}, and our method still has the highest CPSNR in boldface.
For each image, the average execution time requirement (in seconds) of each concerned method is tabulated in the last row of Table \ref{Tab:RGB Table}.
Although our method takes more time than the five traditional methods in {\bf CS}, our method has quality merit. When compared with IDID \cite{Y.Zhang}, our method takes much less time and has higher CPSNR.

Table \ref{Tab:Kw Table} demonstrates the average CPSNR merit of our method relative to the MCIM method \cite{S.Wang}.
The average CPSNR gain of our method is 0.6046 dB over the MCIM method using the bilinear interpolation-based chroma upsampling process at the client side.
In addition, using the bicubic interpolation-based chroma upsampling process, our method also still has higher CPSNR value, as listed in parenthesis in Table \ref{Tab:Kw Table}.
When compared with MCIM, our method takes much less time.

SSIM is used to measure the product of the luminance, contrast, and structure similarity preserving effect between the original image and the reconstructed image.
For $I^{RGB}$, the SSIM value is measured by the mean of the three SSIM values for the R, G, and B color planes.
FSIM is a good image quality assessment with high consistency with the subjective evaluation. FSIM first utilizes the primary feature “phase congruency (PC)” which is contrast invariant and the minor feature “gradient magnitude” to obtain the local quality map, and then FSIM utilizes PC as a weighting function to obtain a quality score. 

 Table \ref{Tab:RGB Table} demonstrates that our method has the highest SSIM and FSIM in boldface among the considered methods.
 Table \ref{Tab:Kw Table} indicates the SSIM and FSIM merits of our method relative to the MCIM method \cite{S.Wang}.

%For $I^{RGB}$, the average execution time requirement (in seconds) of each concerned method for $I^{RGB}$ is tabulated in the last row of Table \ref{Tab:RGB Table} and Table \ref{Tab:Kw Table}. Although our method takes more time than the five methods in {\bf CS}, our method has quality merit. When compared with IDID \cite{Y.Zhang} and MCIM \cite{S.Wang}, our method takes much less time.
                                           
\subsubsection{For $I^{Bayer}$}
Since the Kodak and IMAX datasets are obtained by scanning the films, each scanned RGB full-color image can be taken as a ground truth RGB full-color image.
Therefore, the datasets used for $I^{Bayer}$ are created from the two datasets.
The PSNR, SSIM, and FSIM values are used to compare the quality of the reconstructed Bayer CFA images among the considered methods.Here, the reconstructed Bayer CFA images can be viewed as gray images when we measure their PSNR, SSIM, and FSIM values.

In Table \ref{Tab:CFA Table}, “BILI” denotes the bilinear interpolation-based chroma upsampling process used at the client side, and we consider the best combinations for the three comparative chroma subsampling methods, namely 4:2:0(A), the DI method \cite{C.Lin}, and the GD method \cite{Y.Lee}.
Table \ref{Tab:CFA Table} indicates that the PSNR gains of our combination, which denotes the combination ``our method-BILI", over 4:2:0(A)-BILI, DI-COPY \cite{C.Lin}, and GD-BILI \cite{Y.Lee} are 6.4080 dB, 2.9973 dB, and 1.3114 dB, respectively. Table \ref{Tab:CFA Table} also indicates that our combination has the highest SSIM and FSIM in boldface relative to the three comparative combinations.

The last row of Table \ref{Tab:CFA Table} demonstrates that the execution time (in seconds) comparison. 4:2:0(A) is the fastest chroma subsampling method among the four considered chroma subsampling methods. Our method is faster than the GD method \cite{Y.Lee}, but is slower than the DI method \cite{C.Lin}.

\subsubsection{For $I^{IDID}$}
The DTDI CFA images are collected from the two same datasets.
At the client side, under the BILI based chroma upsampling process, to compare the quality performance among the three considered methods, namely 4:2:0(A), the CD method \cite{W.Yang}, and our method, the three quality metrics, CPSNR, SSIM, and FSIM, are used to measure the quality of the reconstructed DTDI images.
Table \ref{Tab:DTDI Table} indicates that our method has the highest CPSNR, SSIM, and FSIM in boldface among the three considered methods, while the CD method is the fastest.
In particular, the CPSNR gains of our method over 4:2:0(A) and the CD method are 2.1959 dB and 0.9459 dB, respectively.

\begin{figure}[]
  \centering
    \subfigure[]{\includegraphics[height=4cm]{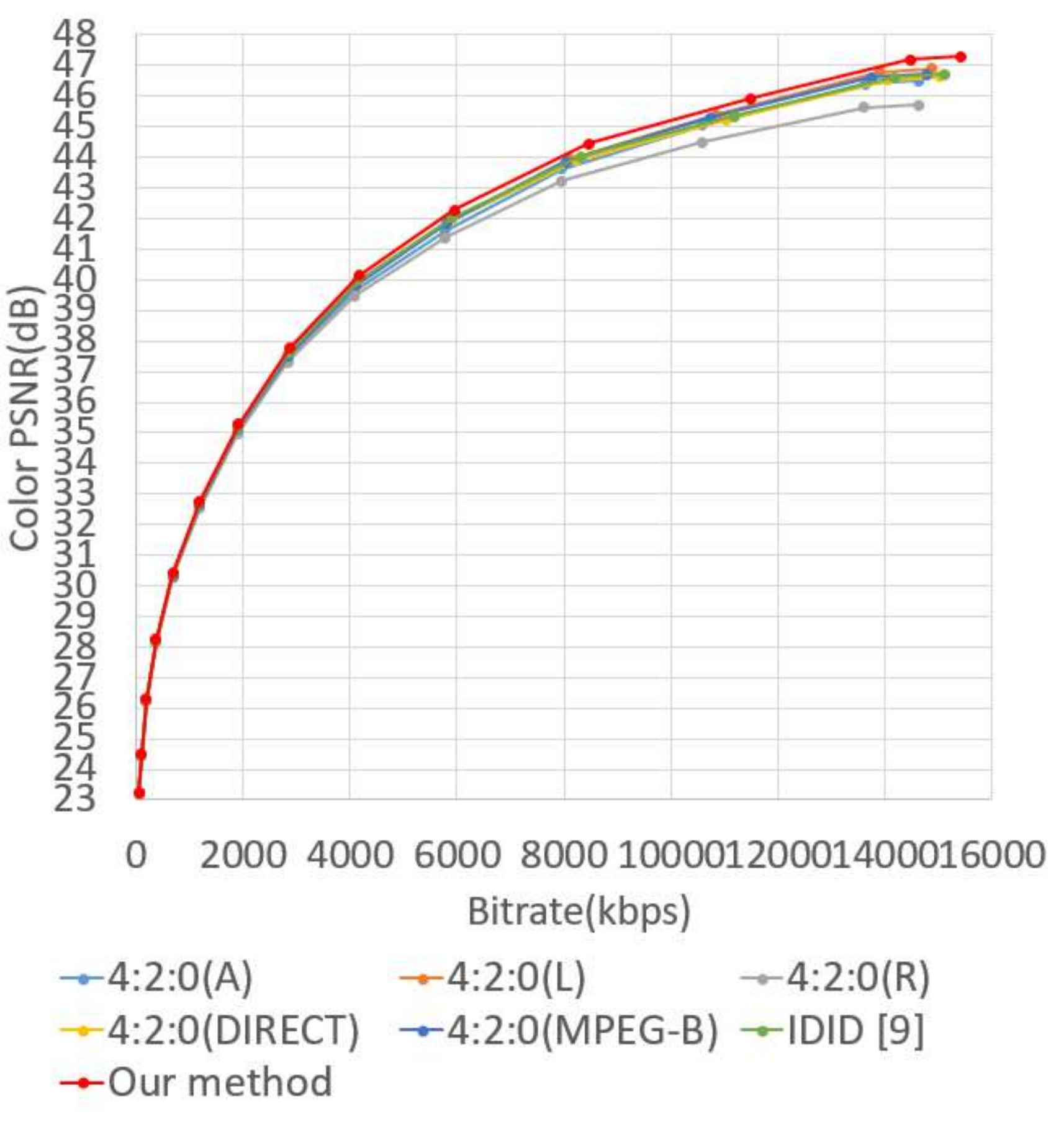}}
    \subfigure[]{\includegraphics[height=4cm]{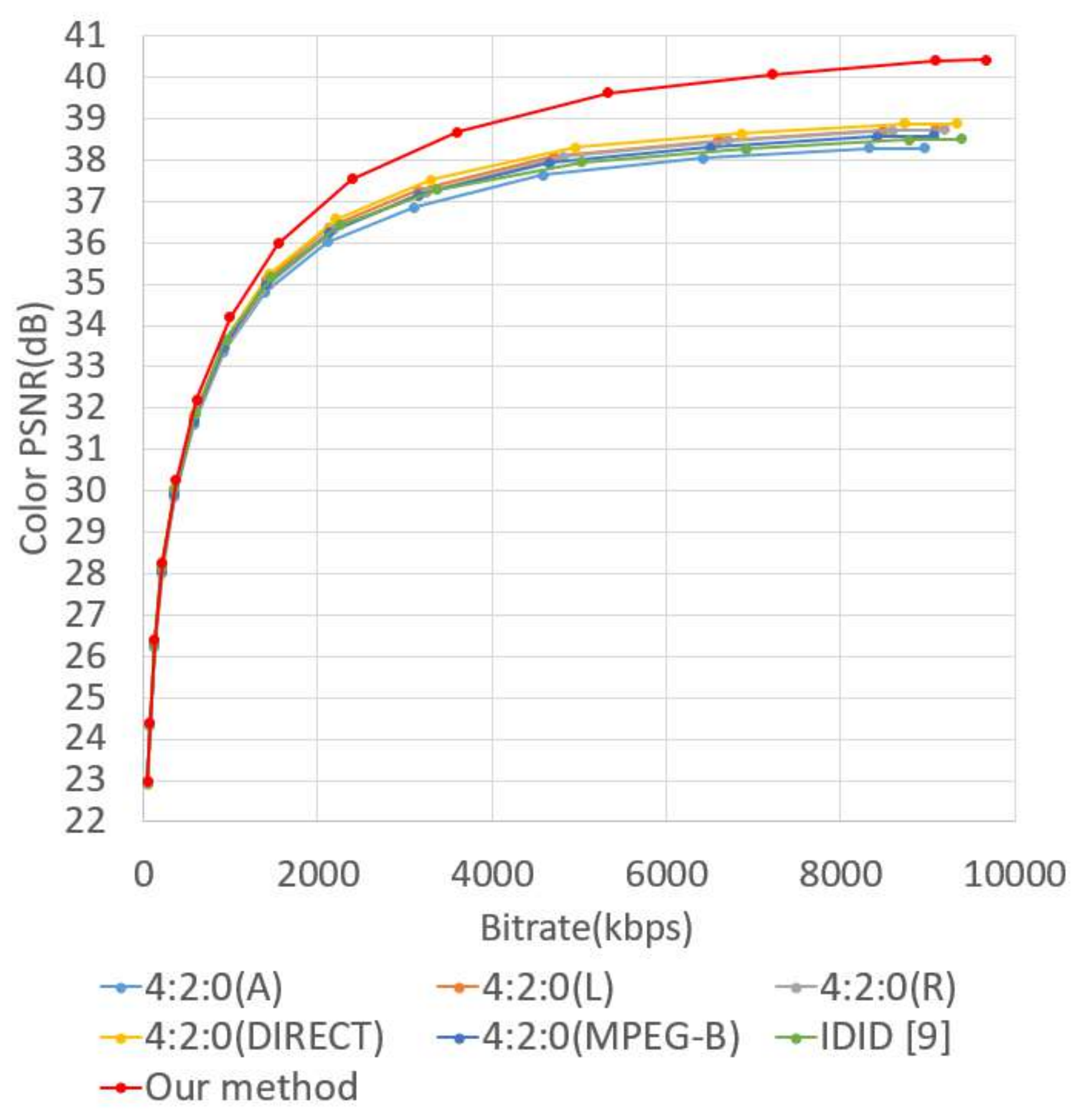}}
    %\subfigure[]{\includegraphics[height=5.5cm]{fig//color-island-vvc.pdf}}
    \caption{Quality-bitrate tradeoff merit of our method for $I^{RGB}$. (a) For Kodak. (b) For IMAX.}
  \label{pic:RD}
\end{figure}
\subsection{Quality-bitrate Tradeoff Merit}
When setting QP = 0, 4, 8, 12, 16, 20, 24, 28, 32, 36, 40, 44, 48, and 51, the quality-bitrate tradeoff merit of our method for $I^{RGB}$, $I^{Bayer}$, and $I^{DTDI}$ is depicted by the RD (rate-distortion) curves for the reconstructed RGB full-color, Bayer CFA, and DTDI CFA images, respectively.
The bitrate of one compressed test set is defined by
\begin{equation} \text{bitrate}=\frac{B}{N} \end{equation}
where $B$ denotes the total number of bits required in compressing the test images.

For $I^{RGB}$, the RD curves corresponding to the Kodak dataset and the IMAX dataset are shown in Fig. \ref{pic:RD}(a) and Fig. \ref{pic:RD}(b), respectively,
in which the X-axis denotes the average bitrate required and the Y-axis denotes the average CPSNR value of the reconstructed RGB full-color images, indicating that under the same bitrate, our method has the highest CPSNR among the considered methods.

For $I^{Bayer}$ and $I^{DTDI}$, the RD curves corresponding to the Kodak and IMAX datasets are depicted in Figs. \ref{pic:Bayer RD}(a)-(b) and Figs. \ref{pic:DTDI RD}(a)-(b), respectively, indicating that under the same bitrate, our method has the best quality relative to the existing methods.

\begin{figure}[]
\centering
    \subfigure[]{\includegraphics[height=4cm]{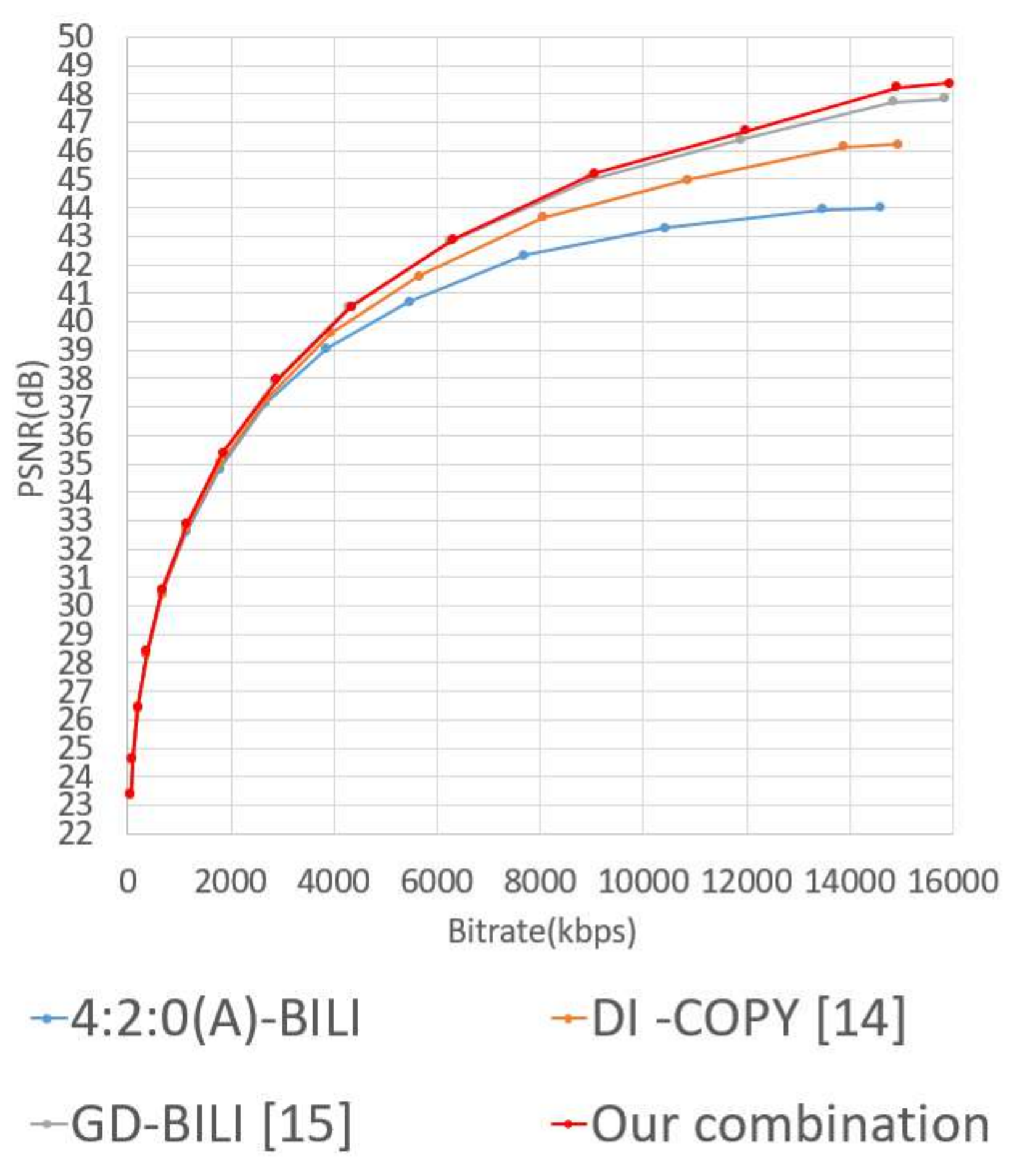}}
    \subfigure[]{\includegraphics[height=4cm]{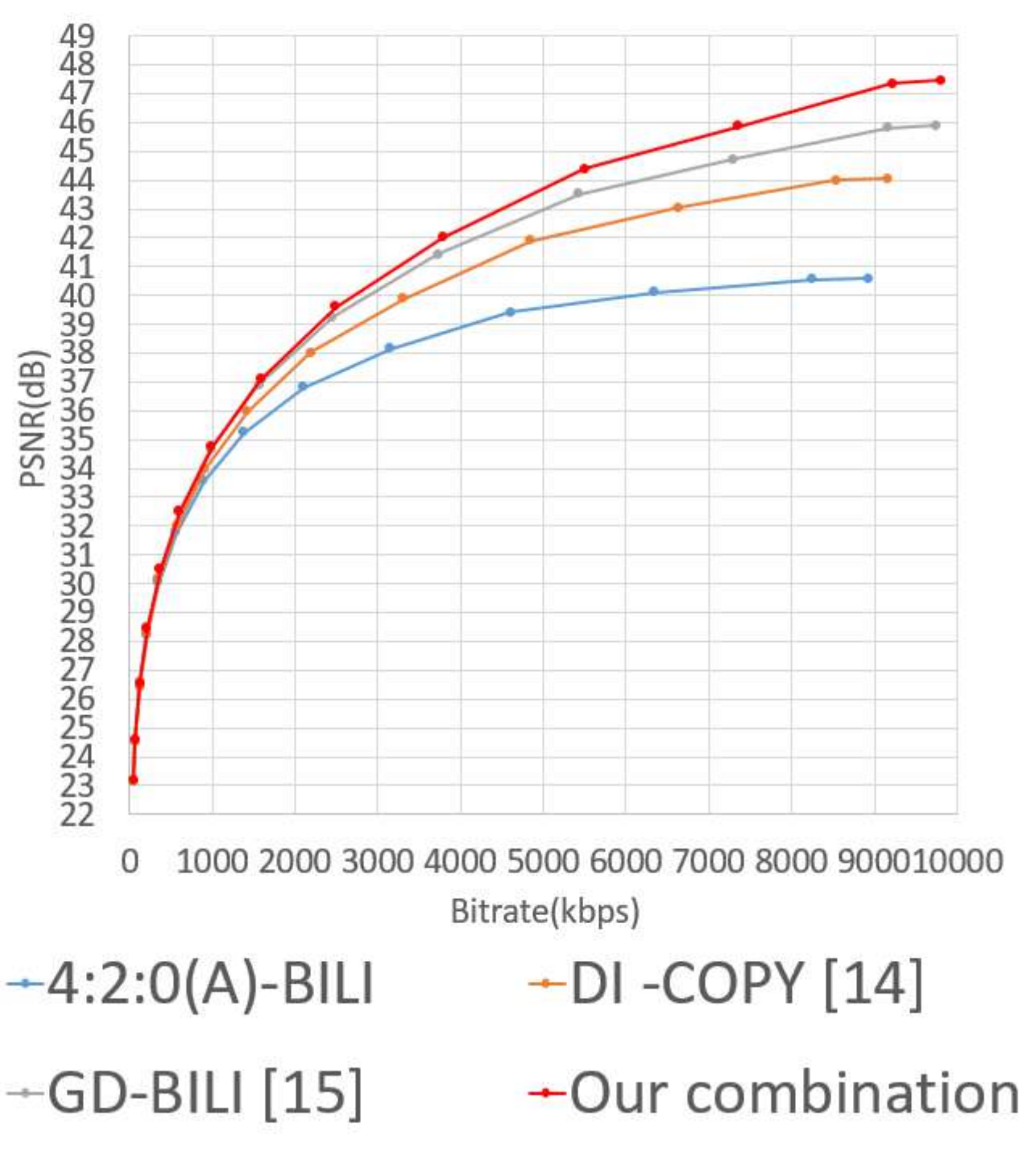}}
    %\subfigure[]{\includegraphics[height=5.5cm]{fig//bayer-island-vvc.pdf}}
  \caption{Quality-bitrate tradeoff merit of our method for $I^{Bayer}$. (a) For Kodak. (b) For IMAX.}
  \label{pic:Bayer RD}
\end{figure}
\begin{figure}
  \centering
    \subfigure[]{\includegraphics[height=4cm]{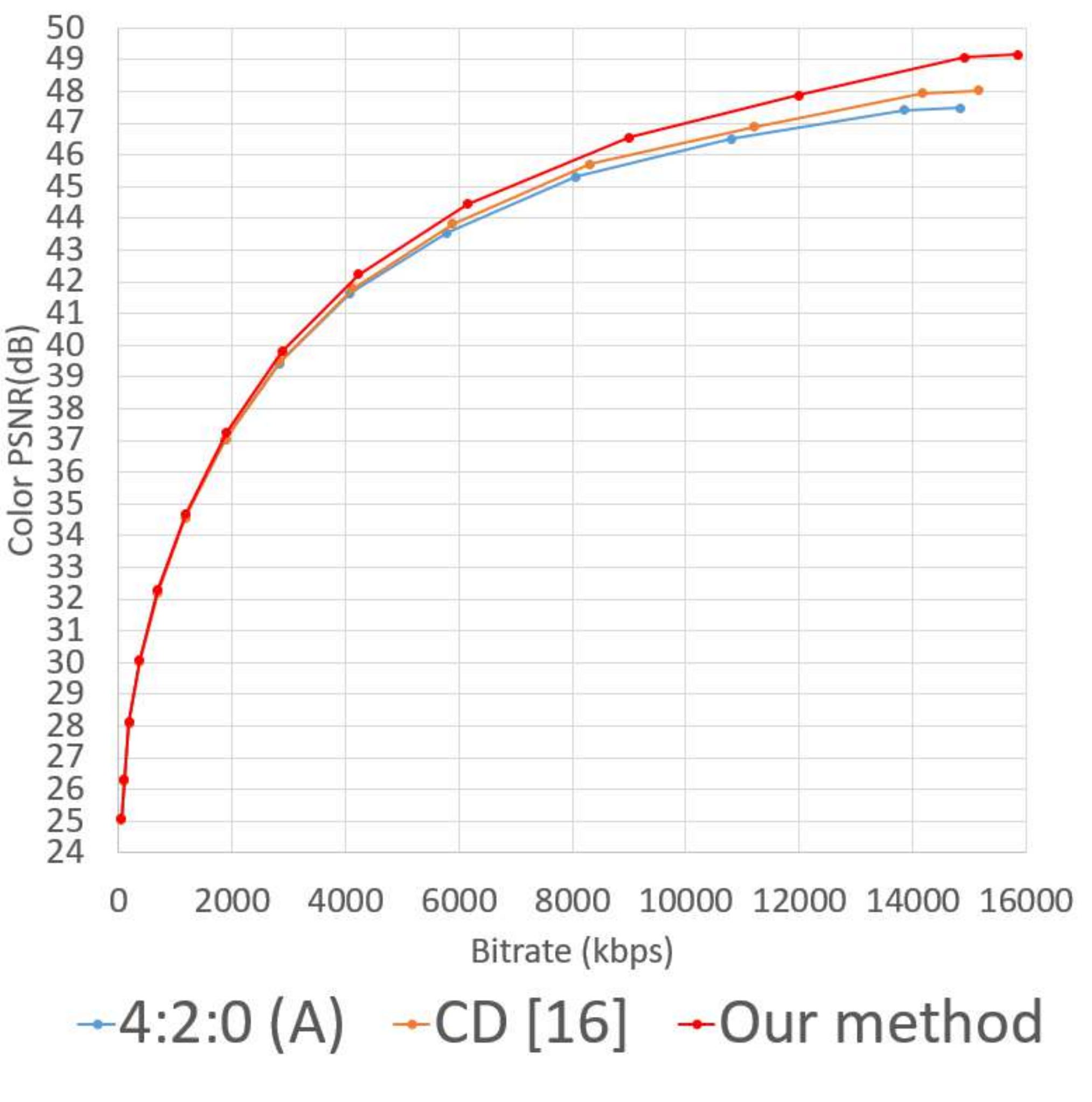}}
    \subfigure[]{\includegraphics[height=4cm]{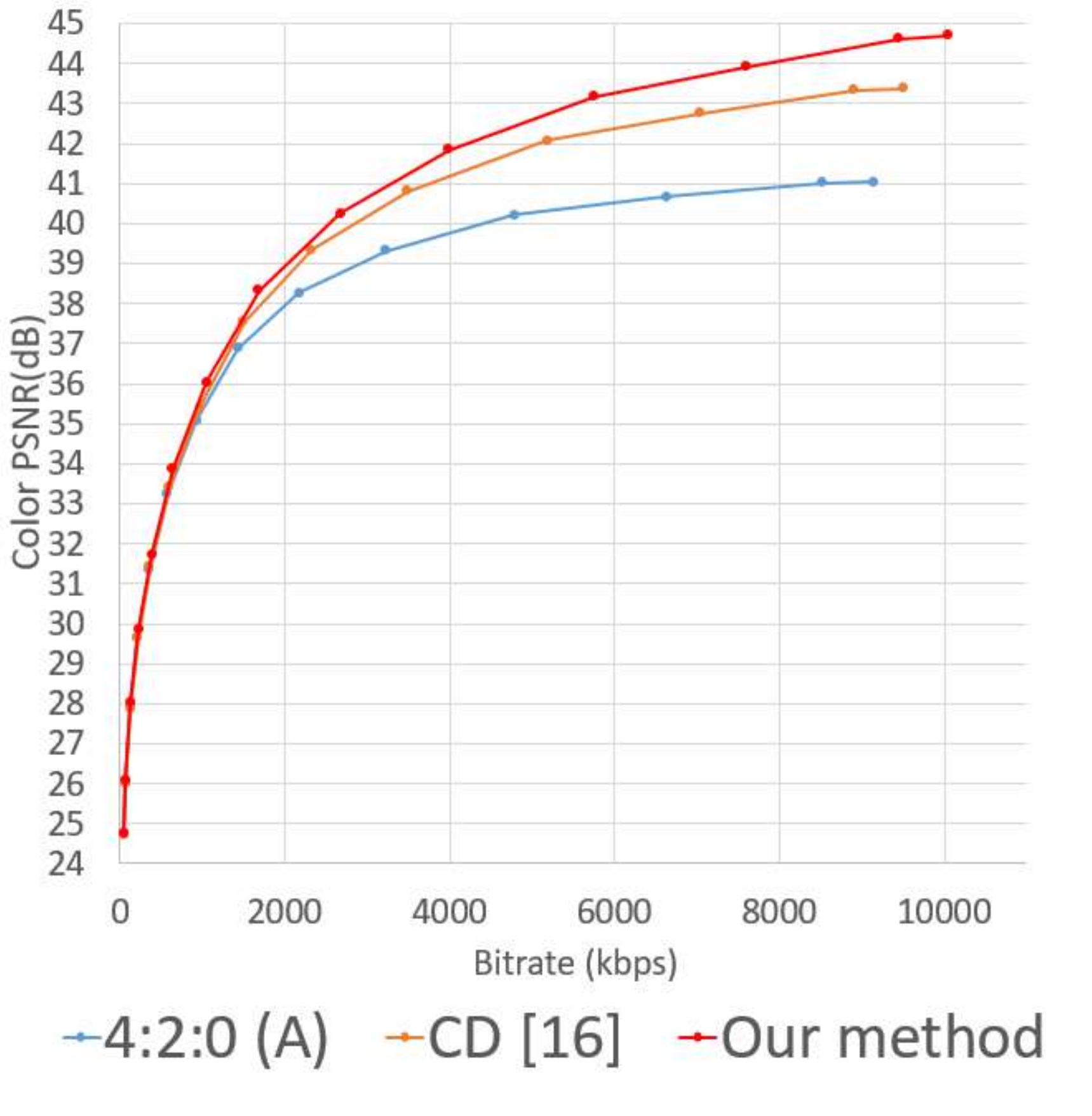}}
    %\subfigure[]{\includegraphics[height=5.5cm]{fig//dtdi-island-vvc.pdf}}
  \caption{Quality-bitrate tradeoff merit of our method for $I^{DTDI}$. (a) For Kodak. (b) For IMAX.}
  \label{pic:DTDI RD}
\end{figure}

%\begin{figure*}
%  \centering
%    \includegraphics[scale=0.4]{fig//Dresden.pdf}
%    \caption{The 30 test RGB full-color images collected from the Dresden dataset.}
%  \label{pic:Dresden}
%\end{figure*}

%\subsection{Application to DTDI images} \label{sec:IVB}

\begin{table}
  \centering
    \caption{THE CPSNR, SSIM, FSIM, AND TIME COMPARISON AMONG THE CONSIDERED METHODS FOR $I^{DTDI}$.}
    \label{Tab:DTDI Table}
  \begin{tabular}{|c|c|c|c|}
\hline
$I^{DTDI}$                 & 4:2:0(A)                 & CD \cite{W.Yang}         & Our method                        \\ \hline
\multirow{2}{*}{CPSNR(dB)} & \multirow{2}{*}{44.2594} & \multirow{2}{*}{45.5094} & \multirow{2}{*}{\textbf{46.4553}} \\
                           &                          &                          &                                   \\ \hline
CPSNR Gain(dB)             & 2.1959                   & 0.9459                   & \textbf{}                         \\ \hline
SSIM (dB)                  & 0.9939                   & 0.9957                   & \textbf{0.9966}                   \\ \hline
SSIM Gain (dB)             & 0.0027                   & 0.0009                   & \textbf{}                         \\ \hline
FSIM (dB)                  & 0.99953                  & 0.99945                  & \textbf{0.99966}                  \\ \hline
FSIM Gain (dB)             & 0.00013                  & 0.00020                  & \textbf{}                         \\ \hline
Execution Time(s)          & 0.0013                   & \textbf{0.0008}          & 0.0291                            \\ \hline
\end{tabular}
 \end{table}

\section{Discussion and Concluding Remarks} \label{sec:VI}

We have presented our improved chroma subsampling method for $I^{RGB}$, $I^{Bayer}$, and $I^{DTDI}$.
First, we propose an improved 2$\times$2 $t$ ($\in \{RGB, Bayer, DTDI\}$) color block-distortion function, and then we prove the convex property of our improved $t$ color block-distortion function.
Next, using the differentiation technique in real domain, the closed form of the initially subsampled $(U, V)$-pair of each 2$\times$2 chroma block $B^{UV}$ is derived.
Furthemore, based on the shape similarity between the convex block-distortion function in real domain and that in the integer domain [0, 255] $\times$ [0, 255], we propose a gradient descent-based iterative method to better improve the subsampled $(U, V)$-pair of $B^{UV}$.
Based on the Kodak and IMAX datasets, the comprehensive experimental results have justified the quality and quality-bitrate tradeoff merits of our improved chroma subsampling method for the above-mentioned three kinds of images relative to the existing traditional and state-of-the-art chroma subsampling methods. 
In fact, given a color image $I^t$ ($\in \{I^{RGB}, I^{Bayer}, I^{DTDI}\}$), after performing our improved chroma subsampling method for $I^t$, we can also apply the chroma subsampling-first luma modification method \cite{chung2017joint} to further improve the quality of the reconstructed $t$ color images.

There are two research issues to be studied in the future. In the first future work, besides the BILI-based chroma upsampling process used to estimate the four $(U, V)$-pairs of $B^{UV}$, which are deployed into our improved 2$\times$2 block-distortion function at the server side, we hope to extend our current BILI-based chroma upsampling approach to the other nonlinear upsampling process to better improve the quality of the reconstructed color images.

In the second future work, we hope to deploy our improved chroma subsampling method into the adaptive chroma subsampling-binding and luma-guided (ASBLG) chroma upsampling method \cite{chung2017adaptive} at the client side, achieving better quality of the reconstructed screen content images.
However, because the subsampling position of our improved chroma subsampling method is exactly the same position as 4:2:0(A), it is very difficult to distinguish our improved method from 4:2:0(A) when using the winner-first voting strategy \cite{chung2017adaptive} at the client side to conclude the accurate correlation between the subsampled decoded luma image and decoded subsampled chroma image.
We also put it into the future work.

\appendix
 \section{THE PROOF OF PROPOSITION \ref{pro1}.}\label{FirstAppendix}

Without loss of generality, we only derive Eq. (\ref{eq:rec U}) for the estimated U component of $B^U$, namely $U'_1$. As depicted in Fig. \ref{pic:3x3Bilinear}, we assume that the coordinates of the three reference-known subsampled U components, namely $U_T$, $U_{TL}$, and $U_L$, are at the locations (1, 1), (0, 1), and (0, 0), respectively.
Therefore, $U'_1$ is located at (3/4, 1/4); the coordinate of the reference-unknown subsampled U component of $B^U$, namely $U_s$, is located at (1, 0).
Applying the bilinear interpolation on the four reference U components to estimate $U'_1$, it yields 
\begin{equation}\label{eq:u1}
  \begin{aligned}
    &{U'_1} = (\frac{3}{4})(1-\frac{1}{4})U_s
    +(1 - \frac{3}{4})(\frac{1}{4})U_{TL}\\
    &\quad\quad+(\frac{3}{4})(\frac{1}{4})U_T
    +(1-\frac{3}{4})(1-\frac{1}{4})U_L\\
    &\quad= \frac{9}{16}U_s+\frac{1}{16}U_{TL}+\frac{3}{16}U_{T}+\frac{3}{16}U_{L}
  \end{aligned}
\end{equation}
Following the above derivation, we can also derive Eq. (\ref{eq:rec U}) for the other three estimated U components of $B^U$, namely $U'_2$, $U'_3$, and $U'_4$.
We complete the proof.

\section{THE PROOF OF PROPOSITION \ref{thm:convex}.}\label{proof}
According to the quadratic convex function definition \cite{K.Binmore}, if a quadratic function is positive definite, it is called a convex function.
Therefore, we prove Proposition \ref{thm:convex} by proving that our quadratic 2$\times$2 $t$ ($\in \{RGB, Bayer, DTDI\}$) color block-distortion function $D^t(U_s, V_s)$ in Eq. (\ref{eq:General BD}) is positive definite.
If so, the quadratic function $D^t(U_s, V_s)$ is called a convex function.

To prove the positive definite property of $D^t(U_s, V_s)$, alternately, we prove that the determinant of the Hessian matrix of $D^t(U_s, V_s)$, denoted by $det(H(D^t(U_s, V_s)))$, is positive \cite{K.Binmore}.
First, the Hessian matrix of $D^t(U_s, V_s)$ is expressed by 
\begin{equation}\label{eq:General H1}
H(D^t(U_s, V_s))=
  \begin{bmatrix} \frac{\partial^2 D^t}{\partial U_s^2} & \frac{\partial^2 D^t}{\partial U_s\partial V_s}\\
                 \frac{\partial^2 D^t}{\partial V_s\partial U_s} & \frac{\partial^2 D^t}{\partial V_s^2}
  \end{bmatrix}
\end{equation}
with
\begin{equation}\label{eq:General H2}
\begin{aligned}
    \frac{\partial^2 D^t}{\partial U_s^2} &= 2w^2\sum\limits_{i=1}\limits^4\sum\limits_{c\in S^{t}_i}a_c^2\\
    \frac{\partial^2 D^t}{\partial V_s^2} &= 2w^2\sum\limits_{i=1}\limits^4\sum\limits_{c\in S^{t}_i}b_c^2\\
    \frac{\partial^2 D^t}{\partial U_s \partial V_s} &=\frac{\partial^2 D^t}{\partial V_s \partial U_s} = 2w^2\sum\limits_{i=1}\limits^4\sum\limits_{c\in S^{t}_i}a_cb_c
\end{aligned}
\end{equation}
Furthermore, the determinant of $H(D^t(U_s, V_s))$ is expressed as
\begin{equation}\label{eq:General detH}
\begin{small}
\begin{aligned}
  &det(H(D^t))\\
  &=4w^4(\sum\limits_{i=1}\limits^4\sum\limits_{c\in S^{t}_i}a_c^2\sum\limits_{i=1}\limits^4\sum\limits_{c\in S^{t}_i}b_c^2 -\sum\limits_{i=1}\limits^4\sum\limits_{c\in S^{t}_i}a_cb_c\sum\limits_{i=1}\limits^4\sum\limits_{c\in S^{t}_i}a_cb_c)\\
  &=64w^4(\sum\limits_{c\in S^{t}_i}a_c^2\sum\limits_{c\in S^{t}_i}b_c^2 - (\sum\limits_{c\in S^{t}_i}a_cb_c)^2).
\end{aligned}
\end{small}
\end{equation}
According to Eq. (\ref{eq:General detH}), we now prove that the value of $H(D^t(U_s, V_s))$, $t \in \{RGB, Bayer, DTDI\}$, 
is always positive for the considered three kinds of images.
Considering $t$ = ``$Bayer$'', by Eq. (\ref{eq:General detH}) and Eq. (\ref{eq:ab}), it yields

\begin{equation}\label{eq:Bayer detH}
    \begin{aligned}
        &detH(D^{Bayer})\\
        &=4w^4(\sum\limits_{i=1}\limits^4\sum\limits_{c\in S^{\text{\tiny $Bayer$}}_i}a_c^2\sum\limits_{i=1}\limits^4\sum\limits_{c\in S^{\text{\tiny $Bayer$}}_i}b_c^2 -(\sum\limits_{i=1}\limits^4\sum\limits_{c\in S^{\text{\tiny $Bayer$}}_i}a_cb_c)^2)\\
        &=4(\frac{9}{16})^4(((0.391)^2+0^2+(2.018)^2+(0.391)^2)\\
        &\quad\times((0.813)^2+(1.596)^2+0^2+(0.813)^2)\\
        &\quad-((0.391\times0.813)+(0\times1.596)\\
        &\quad+(2.018\times0)+(0.391\times0.813))^2)\\
        &=6.6216 > 0.
    \end{aligned}
\end{equation}
We have proved that the value of $det(H(D^{Bayer}))$ is positive. On the other hand, the block-distortion function $D^{Bayer}(U_s, V_s)$ is a convex function. Next, we consider $t$ = ``DTDI'', by Eq. (\ref{eq:General detH}), it yields
\begin{equation}\label{eq:DTDI detH}
  \begin{aligned}
    &detH(D^{DTDI})\\
    &=4w^4(\sum\limits_{i=1}\limits^4\sum\limits_{c\in S^{\text{\tiny $DTDI$}}_i}a_c^2\sum\limits_{i=1}\limits^4\sum\limits_{c\in S^{\text{\tiny $DTDI$}}_i}b_c^2 -(\sum\limits_{i=1}\limits^4\sum\limits_{c\in S^{\text{\tiny $DTDI$}}_i}a_cb_c)^2)\\
    &=4(\frac{9}{16})^4((4\times(0.391)^2+2\times(0^2)+2\times(2.018)^2)\\
    &\quad\times(4\times(0.813)^2+2\times(1.596)^2+2\times0^2)\\
    &\quad-(4\times(0.391\times0.813)+2\times(0\times1.596)\\
    &\quad+2\times(2.018\times0))^2)\\
    &=26.4863 > 0.%17.4628
  \end{aligned}
\end{equation}
We have proved that $det(H(D^{DTDI})) > 0$, and it indicates that our block-distortion function $D^{DTDI}(U_s, V_s)$ is also a convex function.
When $t$ = ``$RGB$'', by Eq. (\ref{eq:General detH}), it can be proved that $det(H(D^{RGB})) = 86.2040 > 0$, indicating the convex property of $D^{RGB}(U_s, V_s)$; for reducing the paper length, we omit the detailed proof for it. We complete the proof. 

\section*{Acknowledgments}
This work was supported by the contracts MOST-$107$-$2221$-E-$011$-$108$-MY$3$ and MOST-$108$-$2221$-E-$011$-$077$-MY$3$ of Ministry of Science and Technology, Taiwan.
The authors appreciate the proofreading help of Ms. C. Harrington to improve the manuscript.

\bibliographystyle{unsrt}  
\bibliography{refs}  %%% Remove comment to use the external .bib file (using bibtex).
%%% and comment out the ``thebibliography'' section.

%%% Comment out this section when you \bibliography{references} is enabled.
%\begin{thebibliography}{1}

%\end{thebibliography}

\end{document}